\documentclass[cleveref,thm-restate,enumerate]{lipics-v2021}
\newcommand\lncs[1]{}
\newcommand\lip[1]{#1}
\usepackage{graphicx}
\usepackage[utf8]{inputenc}
\usepackage{amsmath}
\usepackage{amsfonts}   
\usepackage{booktabs}    
\usepackage{xcolor}      
\usepackage{hyperref}
\usepackage{thmtools} 
\usepackage{apptools}
\lncs{
\usepackage[T1]{fontenc}
\usepackage{thm-restate}
\usepackage[capitalize]{cleveref}
\usepackage{cite}
\usepackage[font=small,labelfont=bf]{caption}
\usepackage[font=small,labelfont=small]{subcaption}
\usepackage[shortlabels]{enumitem}
\AtEndEnvironment{proof}{\phantom{}\qed} }
\usepackage[textsize=small,disable]{todonotes}
\let\qd\qed
\renewcommand\qed{\hfill$\qd$}

\lip{\nolinenumbers}
\lip{\newcommand{\restateref}[1]{\IfAppendix{\hyperref[#1]{$\star$}}{\hyperref[#1*]{$\star$}}}}
\lncs{\newcommand{\restateref}[1]{\href{https://arxiv.org/abs/2502.18390}{$\star$}}}

\title{Unbent Collections of Orthogonal Drawings\lncs{\thanks{The full version of this paper is available on arXiv~\cite{OurArxiv}.}}}

\lncs{
\author{
  Todor Anti\'c \inst{1}\orcidID{0009-0008-6521-7987}\thanks{Supported by Czech Science Foundation grant no.~23-04949X and grant no.\ SVV–2023–260699.}
  \and
  Giuseppe Liotta\inst{2}\orcidID{0000-0002-2886-9694}\thanks{Research  supported in part  by MUR of Italy, PRIN Project no.~2022TS4Y3N~-- EXPAND and PON Project ARS01\_00540.}
  \and
  Tom\'a\v{s} Masa\v{r}\'ik\inst{3}\orcidID{0000-0001-8524-4036}\thanks{Supported by the Polish National Science Centre SONATA-17 grant number 2021/43/D/ST6/03312.}
  \and
  Giacomo~Ortali\inst{2}
  \and
  Matthias~Pfretzschner\inst{4}\orcidID{0000-0002-5378-1694}\thanks{Funded by the Deutsche Forschungsgemeinschaft (DFG, German Research Foundation) -- grant 541433306.}
  \and
  Peter Stumpf\inst{1,5}\orcidID{0000-0003-0531-9769}\thanks{Supported by Czech Science Foundation grant no.~23-04949X}
  \and
  Alexander~Wolff\inst{6}\orcidID{0000-0001-5872-718X}
  \and
  Johannes~Zink\inst{7}\orcidID{0000-0002-7398-718X}
}

\institute{
    Charles University, Prague, Czech Republic\\
    \and
    Universit\`a degli Studi di Perugia, Perugia, Italy\\
    \and
    University of Warsaw, Warsaw, Poland\\
    \and
    Universit\"at Passau, Passau, Germany\\
    \and
    Czech Technical University in Prague, Czech Republic\\
    \and
    Universit\"at W\"urzburg, W\"urzburg, Germany \\
    \and
    Technische Universit\"at M\"unchen, Heilbronn, Germany
}
}

\lip{
\author{Todor Anti\'c}{Charles University, Prague, Czech Republic}{todor@kam.mff.cuni.cz}{https://orcid.org/0009-0008-6521-7987}{Supported by Czech Science Foundation grant no.~23-04949X and grant no.\ SVV–2023–260699.}
\author{Giuseppe Liotta}{Universit\`a degli Studi di Perugia, Perugia,
  Italy}{giuseppe.liotta@unipg.it}{https://orcid.org/0000-0002-2886-9694}{Research  supported in part  by MUR of Italy, PRIN Project no.~2022TS4Y3N~-- EXPAND and PON Project ARS01\_00540.}
\author{Tom\'a\v{s} Masa\v{r}\'ik}{University of Warsaw, Warsaw,
  Poland}{masarik@mimuw.edu.pl}{https://orcid.org/0000-0001-8524-4036}{Supported by the Polish National Science Centre SONATA-17 grant number 2021/43/D/ST6/03312.}
\author{Giacomo Ortali}{Universit\`a degli Studi di Perugia, Perugia,
  Italy}{giacomo.ortali@gmail.com}{}{}
\author{Matthias Pfretzschner}{Universit\"at Passau, Passau,
  Germany}{pfretzschner@fim.uni-passau.de}{https://orcid.org/0000-0002-5378-1694}{Funded by the Deutsche Forschungsgemeinschaft (DFG) -- grant 541433306.}
\author{Peter Stumpf}{Charles University, Prague, Czech Republic\\ Czech Technical University in Prague, Czech Republic}{stumpf@kam.mff.cuni.cz}{https://orcid.org/0000-0003-0531-9769}{Supported by Czech Science Foundation grant no.~23-04949X}
\author{Alexander~Wolff}{Universit\"at W\"urzburg, W\"urzburg,
  Germany}{}{https://orcid.org/0000-0001-5872-718X}{}
\author{Johannes~Zink}{Technische Universit\"at M\"unchen, Heilbronn,
  Germany}{zink@algo.cit.tum.de}{https://orcid.org/0000-0002-7398-718X}{}
}

\authorrunning{Anti\'c, Liotta, Masa\v{r}\'ik, Ortali, Pfretzschner,
  Stumpf, Wolff, and Zink}

\newcommand{\un}{\operatorname{un}}
\newcommand{\tbn}{\operatorname{tbn}}
\newcommand{\opt}{\operatorname{OPT}}
\newcommand{\alg}{\operatorname{ALG}}
\newcommand{\cost}{\operatorname{cost}}
\newcommand{\ext}{\operatorname{ext}}
\newcommand{\short}{\operatorname{short}}
\newcommand{\bad}{\operatorname{bad}}
\lncs{\renewcommand{\orcidID}[1]{\href{https://orcid.org/#1}{\includegraphics[scale=.03]{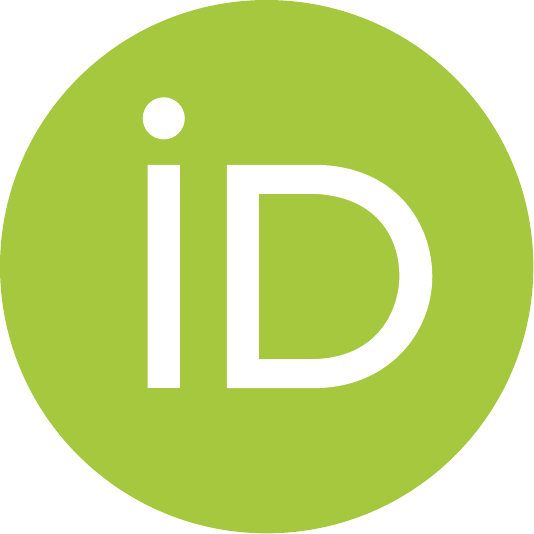}}}}

\lip{\hideLIPIcs}

\lip{}

\graphicspath{{figures/}}
\definecolor{dark blue}{rgb}{0.121,0.47,0.705}
\let\emph\relax\DeclareTextFontCommand{\emph}{\color{dark blue}\em}

\newcommand{\sketchname}{\lip{Proof sketch}\lncs{sketch}}

\lip{\renewcommand{\paragraph}[1]{\subparagraph{#1}}}

\lip{
}

\lip{
}

\lip{
\acknowledgements{We thank the organizers of the Homonolo Workshop
  2023, where this work started.}
}

\begin{document}

\maketitle
\begin{abstract}
  Recently, there has been interest in representing single graphs by multiple
  drawings; for example, using graph stories, storyplans, or uncrossed collections. 
  In this paper, we apply this idea to orthogonal graph drawing.  Due to
  the orthogonal drawing style, we focus on 4-graphs, that is,
  graphs of maximum degree~4.  We restrict ourselves to plane graphs, 
  that is, planar graphs whose embedding is fixed.  Our
  goal is to represent any plane 4-graph~$G$ by an \emph{unbent
  collection}, that is, a collection of orthogonal drawings of $G$
  that adhere to the embedding of $G$ and ensure that each edge of
  $G$ is drawn without bends in at least one of the drawings.  We
  investigate two objectives.
First, we consider minimizing the number of drawings in an unbent
  collection.  We prove that every plane 4-graph can be represented
  by a collection with at most three drawings, which is tight.  We
  also give necessary and sufficient conditions for a graph to admit an unbent
  collection of size~$2$.
  Second, we consider minimizing the total number of bends over all
  drawings in an unbent collection.  We show that this problem is
  NP-hard and give a 3-approximation algorithm.  For the special 
  case of plane triconnected cubic graphs, we show how to compute 
  minimum-bend collections in linear time.
  \lncs{} \end{abstract}

\section{Introduction}\label{se:intro-new}

A way to determine the quality of a drawing of a graph is
its \emph{readability}, which, informally speaking, is the extent to
which a drawing conveys the structure of the graph it represents.
More formally, readability translates into a set of goals that a graph
drawing algorithm must optimize.  For example, in the context of
orthogonal graph drawing, the typical objectives are crossing minimization and bend minimization.  These objectives
typically oppose each other: reducing the number of crossings may
lead to adding bends along the edges, and vice versa.

\paragraph{Related work.}
A recent research direction is to represent a graph by a collection or
sequence of different drawings, each of which optimizes the
readability of a particular subgraph, possibly at the expense of the
rest.  One contribution in this direction is the so-called
\emph{storyplan} approach, where a graph is represented by a sequence
of drawings, called \emph{frames}, each of which only shows a
restricted (e.g., planar) subgraph.  To see the whole (complicated)
graph means to watch all (relatively simple)
frames~\cite{DBLP:journals/jcss/BinucciGLLMNS24,DBLP:conf/sofsem/FialaFLWZ24}.
Another contribution are \emph{graph stories}, where each frame contains the same number of vertices, vertices are shown in an interval of consecutive frames, and, for each vertex, the index of the first frame when it appears is given as input. 
As in a storyplan, it is required that the drawing of an edge and the position of a vertex is the same in every frame where the edge or the vertex appears.  
The task is to compute graph stories with specific properties (e.g., crossing-free frames~\cite{DBLP:journals/jgaa/BattistaDGGOPT23}).
Another approach is to visualize the graph by means of a collection of
drawings, each of which includes all vertices and edges, but to
optimize the drawings only with respect to specific subgraphs.
Hlin\v{e}n{\'{y}} and Masa\v{r}{\'{\i}}k \cite{hm-mucd-GD23} adopted
this idea and introduced the problem of representing a non-planar
graph $G$ via multiple drawings such that $G$ is fully represented in
each drawing and each edge of $G$ is not crossed in at least one of
the drawings.  They called such a collection an \emph{uncrossed
  collection}.  The \emph{size} of a collection is the number
of drawings it contains.  Hlin\v{e}n{\'{y}} and Masa\v{r}{\'{\i}}k
studied two optimization problems.  For a given graph~$G$, the first
problem called \emph{the uncrossed number} is to find a minimum-size uncrossed collection of~$G$.  The
second problem is to find an uncrossed collection that minimizes the
total number of crossings over all drawings, for which they showed an NP-hardness and an FPT algorithm parameterized by the solution size.
Recently, the authors of~\cite{balko-GD24} determined the uncrossed number of complete and complete bipartite graphs exactly.
They also gave a lower bound for the uncrossed number of general graphs in terms of vertex and edge numbers.

\paragraph{Our contribution.}
We revisit the classical problem of computing a planar
orthogonal layout with the minimum number of bends from the
``collection perspective'' of Hlin\v{e}n{\'{y}} and
Masa\v{r}{\'{\i}}k.  Namely, we study the following problem: Given a
plane 4-graph $G$ (i.e., a planar graph with vertex degree at
most $4$ and with a fixed planar embedding), compute a set
of planar orthogonal drawings of $G$ such that every edge of $G$ is
straight (i.e., has no bend) in at least one drawing of the collection.
Analogously to the terminology of Hlin\v{e}n{\'{y}} and
Masa\v{r}{\'{\i}}k, we call such a collection an \emph{unbent collection}. 
We study two problems.
For a given graph~$G$, we want to find an unbent collection of minimum size, $\un(G)$, and an unbent collection that
minimizes the total number of bends, $\tbn(G)$.  Refer to \cref{fig:k4}.

\begin{figure}[t]
\begin{subfigure}{0.28\textwidth}
    \centering
    \includegraphics[page=2]{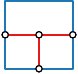}
    \subcaption{} \label{fig:k4-3}
  \end{subfigure}
  \hfill
  \begin{subfigure}{.68\textwidth}
    \centering
    \includegraphics[page=3]{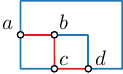}\quad\quad\includegraphics[page=4]{k4-SMALl}
    \subcaption{\lip{\nolinenumbers{}}}
    \label{fig:k4-b}
  \end{subfigure}
  \caption{(a) A minimum-bend orthogonal drawing of~$K_4$.
    (b) An unbent collection
    for $K_4$ with two drawings showing that $\un(K_4)=2$
    and $\tbn(K_4) \le 12$ (in fact, $\tbn(K_4) = 12$); see \cref{def:unbent}.
    Red edges are straight, whereas blue edges are bent.
}
\label{fig:k4}
\end{figure}

Specifically, we show the following.
For every plane 4-graph, an unbent collection of size~3 
always exists (this is tight) and can be found in almost-linear time; see \cref{Thm:unbentnumberlessthan3}. 
We also give a necessary and sufficient condition for graphs to have an unbent collection of size~2.  The condition is based on a new concept of
{\em balanced} 2-edge-coloring that we introduce in \cref{se:unbent}.

Then, we show that it is NP-hard to decide, given a plane 4-graph~$G$
and a positive integer~$k$, whether $G$ admits an unbent collection
with at most $k$ bends; see \cref{se:bends}.
Note that if the embedding of $G$ is not fixed, it is even NP-hard to
decide whether $G$ admits a (single) drawing without
bends~\cite{DBLP:journals/siamcomp/GargT01}.
Next, we give a 3-approximation algorithm for minimizing the number of
bends in an unbent collection.
It relies on Tamassia's flow network~\cite{t-oeggwmnb-sjc}, combined with local
modifications and the insight that every 4-graph has star arboricity
at most~4.

For the special case of plane triconnected cubic graphs, we show how to compute minimum-bend collections in linear time; see \cref{se:sub-cubic}.
The collections that we compute have size~2, which is optimal.
Our algorithm builds upon ideas of Nishizeki and Rahman~\cite{DBLP:journals/jgaa/RahmanNN99,DBLP:journals/jgaa/RahmanNN03},
further refined by Didimo, Liotta, Ortali, and Patrignani~\cite{DBLP:conf/soda/DidimoLOP20}. 

We prove statements marked with a\lip{ clickable} $(\star)$ in the \lncs{full version of this paper~\cite{OurArxiv}}\lip{appendix}.

\paragraph{Preliminaries.}
As shorthand for the set $\{i, i+1, \dots, j-1, j\}$, we use \emph{$[i,j]$}, and for the set $[1, j]$, we use \emph{$[j]$}.
A \emph{planar embedding}, or simply, an \emph{embedding}, of a graph $G$ is the collection of (counter-clockwise) circular orderings of incident edges around every vertex induced by a planar drawing of~$G$. 
We will assume that all of our graphs are given with an embedding and a specified external face; such graphs are called \emph{plane}. 
An \emph{orthogonal drawing} of a graph is a drawing in which every edge is drawn as a sequence of vertical and horizontal line segments. 
Now, we can formally define the main concepts that we study in this paper.

\begin{definition}
  \label{def:unbent}
  Let $G$ be a plane 4-graph. A set of orthogonal drawings of~$G$ (all adhering to the embedding of~$G$) is called an \emph{unbent collection} of~$G$ if every edge of~$G$ is drawn without a bend (i.e., straight) in at least one drawing of the set. 
  The \emph{unbent number} of~$G$, denoted $\un(G)$, is the size of a smallest unbent collection of~$G$.
The \emph{unbent bend number} of $G$, denoted $\tbn(G)$, is the minimum
total number of bends, taken over all unbent collections of~$G$.
\end{definition}

Given a plane 4-graph~$G$, Tamassia~\cite{t-oeggwmnb-sjc}
defined a flow network~$N_G$ 
that has a node for each vertex and for each face of~$G$,
an arc from every face to each incident vertex,
and two parallel arcs (one in each direction) for each edge~$e$
of~$G$ connecting the faces incident to~$e$. 
A unit of flow in~$N_G$ represents a $90^\circ$ angle
at a vertex of~$G$ or a bend on an edge of~$G$.
For the sake of presentation, we invert the usual flow direction in this paper.
Then, each vertex has a demand of~$4$, each inner face~$f$ has a demand of $4-\deg(f)$, and the external face $f_{\ext}$ has a demand of $-4-\deg(f_{\ext})$.
In the case of negative demand (i.e., flow excess),
we also speak of \emph{supply} (of flow).
Faces of degree greater than~4 and the external face supply flow,
while vertices and faces of degree~3 demand flow.
Flow in $N_G$ has
non-zero cost only for arcs between faces and vertices; otherwise, the cost is one for
each unit of flow. The total cost equals the number of bends on the
edges of~$G$ in the resulting \emph{orthogonal representation}~\cite{t-oeggwmnb-sjc},
which refines the given embedding by specifying angles and bends (but not the lengths of edges).
Such a representation can then be turned into an orthogonal drawing in time linear in the number of vertices and bends~\cite{t-oeggwmnb-sjc}.

Recently, a minimum-cost flow algorithm that runs in $O(m^{1+o(1)}\log U \log C)$ 
time was presented~\cite{DBLP:conf/focs/Brand0PKLGSS23},
where $m$ is the number of edges, $U$ is the maximum edge capacity,
and $C$ is the maximum absolute edge cost in the flow network.
In~$N_G$, $C \in O(1)$, $m \in O(n)$, and $U \in O(n)$, where $n=|V(G)|$.
Since $\log n = n^{\log\log n / \log n} \in n^{o(1)}$,
the algorithm runs in $O(n^{1 + o(1)})$ time on input~$N_G$.

\section{Minimizing the Size of an Unbent Collection}
\label{se:unbent}

In this section, we give, for any plane 4-graph $G$, an upper bound on $\un(G)$ in terms of the arboricity of $G$ by slightly amending Tamassia's flow network~\cite{t-oeggwmnb-sjc}.
\begin{proposition}\label{prop:unbentSpanningTrees}
    If a plane 4-graph $G$ can be partitioned into $k$ forests, then $\un(G)\le k$.
    Given $G$ and such a set of forests, an unbent collection
    of size~$k$ can be computed in $O(kn^{1 + o(1)})$ time, where $n$ is the number of vertices of~$G$.
\end{proposition}
\begin{proof}
    For $i \in [k]$, let $F_i$ be the $i$-th given forest.
    We construct Tamassia's flow network for~$G$.
    We then make $k$ copies of the network. In the $i$-th copy, for each edge $e \in F_i$, we delete the two arcs that connect the two nodes representing the incident faces of $e$.
    Since $F_i$ contains no cycle, note that the nodes in the flow network representing the faces
    still form a strongly connected component, hence there always exists a solution.
    Solving these $k$ flow networks individually produces $k$ orthogonal drawings.
    Each edge~$e$ of~$G$ is drawn straight in at least one of these drawings, namely in the $i$-th drawing if $e$ lies in $F_i$.
For each ${i \in [k]}$, constructing the $i$-th flow network and
    removing the links corresponding to edges in~$F_i$
    can be done in linear time,
    and each flow network can be solved in $O(n^{1 + o(1)})$ time~\cite{DBLP:conf/focs/Brand0PKLGSS23}.
\end{proof}

\cref{prop:unbentSpanningTrees} and 
Schnyder woods \cite{DBLP:conf/soda/Schnyder90} yield an upper bound on $\un(G)$.

\begin{theorem}
    \label{Thm:unbentnumberlessthan3}
    Any plane 4-graph $G$ admits an unbent collection
    of size at most~$3$.  Such a collection can be computed
    in $O(n^{1 + o(1)})$ time, where $n = |V(G)|$.
\end{theorem}

\begin{proof}
    Extend $G$ to a plane triangulation~$G^+$, and decompose $E(G^+)$ into three disjoint trees in linear time by Schnyder's algorithm~\cite{DBLP:conf/soda/Schnyder90}.
    Ignoring the edges from $G^+ \setminus G$ leaves three edge-disjoint forests,
    to which we apply \cref{prop:unbentSpanningTrees}.
\end{proof}

\cref{Thm:unbentnumberlessthan3} tells us that for any planar graph $G$ of maximum degree $4$, we have $\un(G)\le3$.
We will later show that this bound is in fact tight, by constructing an infinite family of graphs with unbent number $3$.
Using Tamassia's flow network~\cite{t-oeggwmnb-sjc}, we can decide efficiently whether $\un(G)=1$.
Thus, the real challenge is to decide whether two drawings are sufficient or three are necessary.
We give a partial answer to this question below.

\begin{proposition}
    \label{prop:DensityNashWilliams}
Let $G$ be a plane 4-graph.
    If, for every subset $V' \subseteq V(G)$,
    the induced subgraph $G[V']$ has at most $2|V'| -2$ edges,
    then $\un(G) \leq 2$.
\end{proposition}
\begin{proof}
    According to Nash-Williams' theorem~\cite{n-edstfg-JLMS61},
    a graph $H$ can be partitioned into $k$ forests
    if and only if, for every $U' \subseteq V(H)$, the induced graph
    $H[U']$ has at most $k \cdot (|U'| - 1)$ edges.
    Applying Nash-Williams' theorem and \cref{prop:unbentSpanningTrees} to~$G$ and $k = 2$
    yields $\un(G) \leq 2$.
\end{proof}

Testing the sufficient criterion presented in \cref{prop:DensityNashWilliams}
naively takes exponential time.
So, we next present a way to decompose~$G$ into relevant subgraphs.

\begin{restatable}[\restateref{lm:decompose}]{lemma}{Decompose}
    \label{lm:decompose}
    Let $G$ be a plane 4-graph with $n$ vertices and at most $2n - 2$ edges.
    In $O(n)$ time,
    we can construct a set $\mathcal{G}$ of plane 4-graphs such that
    (i)~for each $G' \in \mathcal{G}$, $|E(G')| \ge 2 |V(G')| - 1$,
    (ii)~$\sum_{G' \in \mathcal{G}} |V(G')| \in O(n)$, and
    (iii)~$\un(G) \le 2$ if and only if, for each $G' \in \mathcal{G}$, $\un(G') \le 2$.
\end{restatable}

\begin{proof}[\sketchname]
    Essentially, we decompose $G$ into induced subgraphs
    such that each such subgraph~$G'$ has at least $2|V(G')| - 1$ edges.
    (The other parts of $G$ can be ignored because, due to \cref{prop:DensityNashWilliams},
    they admit an unbent collection of size~2 anyway.)
    Since $G$ has maximum degree 4 and $G'$ misses at most one edge to be 4-regular, $G'$ is connected to $G \setminus G'$ with at most two edges.
    Hence, we find all 1- and 2-edge cuts in linear time,
    which decompose the graph into smaller graphs.
In each smaller graph, we keep the edges of the cuts,
    which would end at leaves (i.e., vertices of degree~1).
    To still match $|E(G')| \ge 2 |V(G')| - 1$ for each smaller graph~$G'$,
    we attach the gadget~$H$ shown in \cref{fig:Deg1Gadget} to each such leaf.
    If $\un(G') \le 2$ holds for each $G'$, we obtain an unbent collection of size $2$ of $G$ by merging the smaller graphs at the leaves of the corresponding cuts.
    Since each cut has size at most $2$, we can keep one edge straight in each of the two drawings and add as many bends as necessary to the other edge.
\end{proof}

\begin{figure}[tb]
    \begin{subfigure}{.3\textwidth}
        \centering
        \includegraphics[page=1]{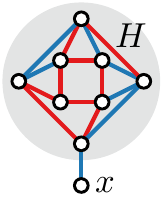}
        \subcaption{}
    \end{subfigure}
    \hfill
    \begin{subfigure}{.65\textwidth}
        \centering
        \includegraphics[page=2]{Deg1Gadget-SMALL}
        \subcaption{}
    \end{subfigure}
    \caption{(a) The gadget $H$ with $2n-1$ edges used to replace a degree-1 vertex in the proof of Lemma~\ref{lm:decompose}.
    (b)~An unbent collection of size 2 for $H$.}
    \label{fig:Deg1Gadget}
\end{figure}

\cref{lm:decompose} tells us that we should focus on graphs with at least $2n-1$ edges
(a 4-graph has at most $2n$ edges, so these are the ``densest'' plane 4-graphs)
as the calculation of the unbent number of all graphs with lower density boils down to these graphs.
It is natural to ask whether graphs with $\un(G)=3$ even exist,
which we answer positively later.
If~$G$ has at least $2n-1$ edges,
then \Cref{prop:unbentSpanningTrees} cannot be applied to obtain $\un(G)\le 2$
as every 2-edge-coloring will contain a cycle in at least one color.
So, we next characterize graphs of unbent number at most~2 via  2-edge-colorings with specific properties.

We say that every vertex of $G$ has $4 - \deg(v)$ \emph{free angles}.
An \emph{angle assignment} of $G$ assigns each free angle of every vertex to one of its incident faces.
Given an angle assignment, the \emph{demand} of a face $f$ of $G$ is the number of angles it has been assigned plus $(4 - \deg(f))$ (respectively $-4 - \deg(f)$ if $f$ is the external face). 
These demands of the faces correspond to a simplified version of
Tamassia's flow network \cite{t-oeggwmnb-sjc} where the flow of the nodes to the faces
has already been assigned and all that still needs to be done
is solving the flow network induced by the nodes representing the faces of~$G$.
A 2-edge-coloring of $G$ is \emph{balanced} if, for each of the two colors, there exists an angle assignment such that, for every cycle of that color, the demands of all faces it encloses sum up to~0.
Using Tamassia's flow network \cite{t-oeggwmnb-sjc}, we can show that the graphs that admit an unbent collection of size at most 2 are exactly the graphs that admit a balanced 2-edge-coloring.

\begin{restatable}[\restateref{the:balancedColoringCharacterization}]{proposition}{BalancedColoring}
\label{the:balancedColoringCharacterization}
    Let $G$ be a plane 4-graph.
    Then $\un(G) \leq 2$ if and only if~$G$ admits a balanced 2-edge-coloring.
\end{restatable}

With this characterization, we can show that graphs with unbent number~3 exist. 
We define the \emph{flower graph} with $k$ petals, $F_k$, to be the graph obtained from two cycles of length $k$, $C=(v_1,\dots,v_k)$ and $C'=(v'_1,\dots, v'_k)$ by adding edges $\{ v_i,v'_i \}$ and $\{v_i,v'_{i+1}\}$ for each $i \in [k]$. For an illustration, see \cref{fig:sunflowers}.
Note that $F_3$ is the octahedron.

\begin{figure}[t]
\begin{minipage}[b]{.57\textwidth}
  \begin{subfigure}{.32\textwidth}
    \centering
    \includegraphics[page=1]{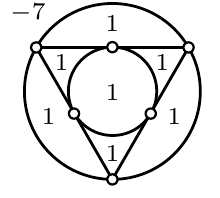}
    \subcaption{\lip{\nolinenumbers{}}}
  \end{subfigure}
  \hfill
\begin{subfigure}{.32\textwidth}
    \centering
    \includegraphics[page=4]{sunflower-SMALL}
    \subcaption{\lip{\nolinenumbers{}}}
    \label{fig:8sunflower}
  \end{subfigure}
\hfill
  \begin{subfigure}{.32\textwidth}
    \centering    
    \includegraphics[page=3]{counterexamples-SMALL}
    \subcaption{\lip{\nolinenumbers{}}} 
    \label{fig:counterexample2}
  \end{subfigure}

  \caption{(a), (b) Flower graphs $F_3$ and $F_8$ with 3 and 8 petals, respectively.
    Note that $\un(F_3) = 3$, whereas $\un(F_8)=2$.
    The 2-edge-coloring shown for $F_8$ is balanced;
    see the demands (the numbers inside the faces).
    (d)~Another 4-regular graph~$G$ with $\un(G)=3$.
}
    \label{fig:sunflowers}
\end{minipage}
\hfill
\begin{minipage}[b]{.39\textwidth}
  \begin{subfigure}[b]{.49\textwidth}
    \centering
    \includegraphics[page=1]{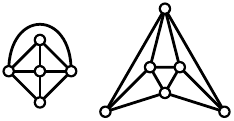}
    \subcaption{}
    \label{fig:minimalCounterexamples}
  \end{subfigure}
  \hfill
  \begin{subfigure}[b]{.49\textwidth}
    \centering
    \includegraphics[page=2]{counterexamples-SMALL}
    \subcaption{}
    \label{fig:infiniteCounterexamples}
  \end{subfigure}
  \caption{(a) Two minimal graphs that require three drawings. (b)~By extending the graphs from~(a), one can construct infinite families of graphs that require three drawings.}
\end{minipage}
\end{figure}

The following \cref{pro:counterexamples} shows that $K_5$ minus an edge and the octahedron minus an edge have unbent number~3 if the large face is the external face; see~\cref{fig:minimalCounterexamples}.
Thus, plane 4-graphs that contain one of those two as an induced (plane) subgraph also have unbent number 3;
see~\cref{fig:infiniteCounterexamples}.
It also shows ${\un(F_k)=3}$ for $k \in [3, 7]$,
and that the 4-regular graph in~\cref{fig:counterexample2}
has unbent number~3.
We remark that $\un(F_k)=2$ for $k \ge 8$; see~\cref{fig:8sunflower}.

\begin{restatable}[\restateref{pro:counterexamples}]{proposition}{Counterexamples}
  \label{pro:counterexamples}
  Let $G$ be a plane 4-graph with $n$ vertices and $m$ edges such that
(i)~$m=2n-1$ and every inner face is a triangle or
(ii)~$m=2n$, all but the external face and exactly one inner face $f$ are triangles and 
$f$ shares an edge with the external face or $f$ has at most seven edges.   
Then, $\un(G)=3$.   
\end{restatable}

\begin{proof}[\sketchname]
Suppose $\un(G)\le 2$.  Then, by~\cref{the:balancedColoringCharacterization}, $G$ admits a balanced 2-edge-coloring.
  Let us first assume condition~(i). The edges of one color must induce a cycle since $m=2n-1$.
  This cycle, however, encloses only triangles and thus cannot be balanced. Hence, $\un(G)=3$.

  Now assume condition~(ii). First suppose that one color induces two cycles $c_1$ and $c_2$
  and let~$F_1$ and~$F_2$ be the sets of faces
  enclosed by $c_1$ and~$c_2$, respectively.
  Because both cycles are balanced, both $F_1$ and $F_2$ contain~$f$.
  The faces in the symmetric difference $F_1 \triangle F_2$ are only triangles
  that are enclosed by at least one monochromatic cycle, which cannot be balanced.
  Hence, the edges of each color induce a spanning graph with exactly one (balanced) cycle containing $f$.
  This is not possible if $f$ shares an edge with the external face. Otherwise, let $k$ be the number of edges of $f$.
  Then each edge~$e$ of $f$ is incident to another triangle that is contained in the cycle of the color not assigned to $e$.
  This results in at least $k$ triangles inside of both monochromatic cycles,
  while, for $k \le 7$, there should be only $2(\deg(f)-4)=2k-8<k$ triangular faces to account for
  the demand of~$f$.  
\end{proof}

\section{The Unbent Bend Number of Plane 4-Graphs}
\label{se:bends}

In this section, we first show that determining the unbent bend number is NP-complete.
Afterwards, we present a 3-approximation algorithm for computing the unbent bend number of any graph that admits a crossing-free orthogonal 
drawing, that is, of any plane 4-graph. 

As a new NP-complete version of the \textsc{SAT} problem,
we introduce \textsc{LevelPlanarMonotone2in4Sat}, which is defined as follows.
As in \textsc{2in4Sat}, the input is a Boolean formula in conjective normal form where every clause consists of four literals, 
and the task is to decide whether an assignment of variables to Boolean values
exists such that, in every clause, exactly two literals are true.
Additionally, the variable--clause incidence graph is planar and admits a drawing where the variable vertices are
placed on a set of parallel lines, called \emph{levels}, and the clause vertices are
placed between the levels such that no edge crosses a level.
The levels divide the plane into rectangular strips.
The clauses that are represented in the first, third, etc.\ strip contain only negated literals, while
the clauses represented in the second, fourth, etc.\ strip contain only unnegated literals.

\begin{restatable}[\restateref{lm:level2in4satNPh}]{lemma}{LevelTwoInFourSAT}
  \label{lm:level2in4satNPh}
  \textsc{LevelPlanarMonotone2in4Sat} is NP-complete.
\end{restatable}

This is the base problem for our NP-hardness reduction sketched next.

\begin{restatable}[\restateref{th:nph}]{theorem}{NPh}
    \label{th:nph}
    Given a plane 4-graph $G$ and an integer~$k$,
    it is NP-complete to decide whether $\tbn(G) \le k$.
\end{restatable}

\begin{proof}[\sketchname]
    We sketch only the NP-hardness.
    Given a \textsc{LevelPlanarMonotone2in4Sat} instance $\Phi$
    with $n$ variables, $m$ clauses, and $\ell$ levels,
    we construct a 4-plane graph~$G$ that follows the level structure
    of the variable--clause incidence graph
    (see \cref{fig:nph-construction} for one level)
    and has an unbent collection of size~2 with at most $k = 2nC$ bends,
    where $C = 2m + \ell$,
    if and only if $\Phi$ is satisfiable.
    
    \begin{figure}[t]
        \centering
        \includegraphics[page=1]{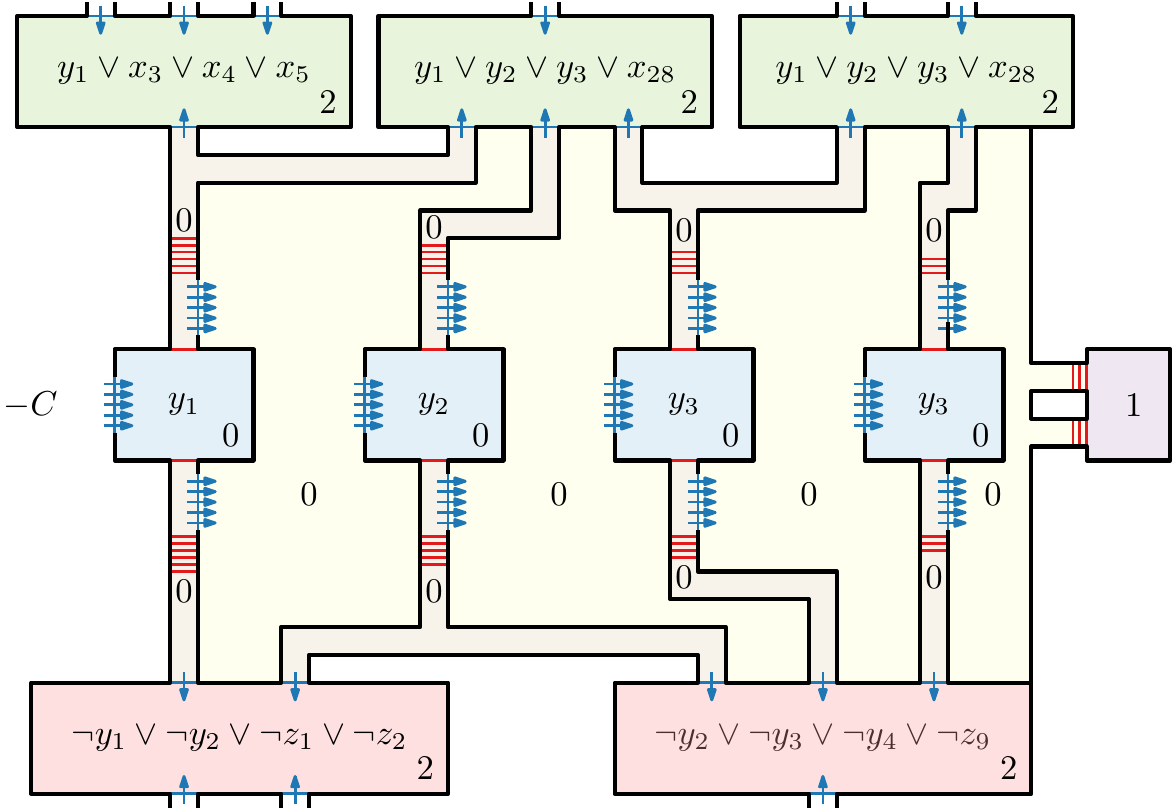}
        \caption{NP-hardness: A layer of variable gadgets and the neighboring clause gadgets.
}
        \label{fig:nph-construction}
    \end{figure}

    In $G$, there is a face for each variable vertex (\emph{variable gadget}; blue in the middle row of \cref{fig:nph-construction}),
    for each clause vertex (\emph{clause gadget};
    green and red in the top and bottom row of \cref{fig:nph-construction}),
    and for each set of edges connecting a variable vertex
    to its clause vertices above and, separately, below
    (\emph{tunnel gadget}; orange and thin in \cref{fig:nph-construction}).
    Each such gadget is surrounded by a set of faces called \emph{walls},
    which make it costly for flow to enter the gadget
    unless it is through some \emph{door}.
    (All walls have ``thickness'' $2n$, that is,
    each unit of flow traversing a wall bends at least $2n$ edges.)
    We construct each face with a specific demand.
    The clause gadgets have a demand of~2,
    the \emph{level gadgets}, which are special faces existing once per level
    (light purple on the very right in \cref{fig:nph-construction}), have a demand of~1,
    the external face has a supply of~$C$
    (shown as a demand of~$-C$ in \cref{fig:nph-construction}),
    and all other faces have a demand (and supply) of~0.
    The demands sum up to~$C$.

    The blue arrows in \cref{fig:nph-construction} indicate \emph{valve gadgets}.
    A valve gadget allows a specific amount of flow
    to be passed from one face (source)
    to an incident face (target) without creating bends;
    passing in the reverse direction is only possible with extra bends.
    Without traversing walls, flow can enter a variable gadget at a door
    being a valve gadget (on the left side of each variable gadget in \cref{fig:nph-construction}),
    and flow can leave a variable gadget at two specific edges (\emph{exit doors};
    red on the top and the bottom side of each variable gadget in \cref{fig:nph-construction}).
    Bending the one exit door represents assigning the variable to true,
    and bending the other exit door represents assigning the variable to false.
    Each exit door can be bent in only one drawing of the desired unbent collection (of size~2).
    
    Any path that a unit of flow can take from the external face
    to a face with positive demand,
    i.e., a clause gadget or a level gadget,
    causes at least~$n$ bends.
    If such a path has exactly $n$ bends,
    it traverses variable gadgets and tunnel gadgets, but no walls.
    Since every variable gadget has one door to enter and
    two exit doors to leave the gadget (the latter being single edges),
    all such paths in the first drawing take the same exit door,
    while all such paths in the second drawing take the other exit door.
    Observe that the level gadget assures that in both drawings,
    at least one unit of flow must take such a path
    starting in the external face and traversing all variable gadgets.
    The choices of exit doors at the variable gadgets
    correspond to an assignment of Boolean values to variables.
    Note that, in \cref{fig:nph-construction},
    there are additional red edges inside the tunnel gadgets.
    Their purpose is to make any path that a unit of flow
    can take from the external face to a clause gadget
    cost $n$ (and not less).
    Moreover, they assure that the clause gadgets
    can only be reached via such a tunnel in one of the drawings.
    
    Each clause gadget has four doors, each being a valve
    gadget allowing one unit of flow to enter.
    These doors are connected to tunnel gadgets,
    which, in turn, are connected to the variable gadgets of
    the four variables that occur in the clause.
Two units of flow in a clause gadget need to come from
    two of these variables gadgets in the one drawing,
    while, in the other drawing, two units of flow come from the other
    two variables gadgets.
Hence, an unbent collection exists if and only if the exit doors taken by the flow paths correspond to a satisfying truth assignment of~$\Phi$.\end{proof}

Next, we describe a 3-approximation algorithm for the unbent bend
number.  We use the property that the star arboricity of a graph is
bounded by its maximum degree.  Algor and Alon~\cite{aa-sag-DM89}
showed that every $\Delta$-regular graph has star arboricity at most
$\Delta/2+O(\Delta^{2/3}(\log \Delta)^{1/3})$.  We are interested only
in the case $\Delta=4$, and give the following simple explicit
construction (stated for general~$\Delta$, although it is
asymptotically worse than the bound of Algor and Alon).

\begin{restatable}[\restateref{prop:arboricity}]{proposition}{Arboricity}
    \label{prop:arboricity}
    Every graph $G$ of maximum degree~$\Delta$ has star arboricity at most~$\Delta$.
    The corresponding partition of $E(G)$ into star forests can be computed in $O(\Delta \cdot |V(G)|)$ time.
\end{restatable}

For every edge $e$ of $G$, let $f_e$ and $g_e$ be the faces that are
adjacent to~$e$ (possibly $f_e=g_e$).  In Tamassia's flow network
$N_G$, there are arcs $(f_e,g_e)$ and $(g_e,f_e)$ of unit cost between
these faces.  The amount of flow on these arcs corresponds to the
number of bends that $e$ makes into the appropriate directions.
Without loss of generality, we can assume that at most one of the two
arcs carries flow; let this arc be~\emph{$e^+$}.  If none carries flow, the choice of~$e^+$ is arbitrary.

\begin{lemma}
  \label{lem:change-flow}
Let $G$ be a plane 4-graph, and let $\phi$ be a feasible flow
  in~$N_G$.  If~$v$ is a vertex of~$G$ that is incident to $d \ge 2$
  edges $e_1,\dots,e_d$ in this order in the given embedding, then
  there exists a feasible flow~$\phi'$ in~$N_G$ such that
  $\phi'(e_1^+)=0$, for $i \in \{2,\dots,d\}$, it holds that
  $\phi'(e_i^+) \le \phi(e_i^+)+\phi(e_1^+)$, and, for every
  arc~$a \not\in \{e_1^+,\dots,e_d^+\}$ of~$N_G$, it holds that
  $\phi'(a)=\phi(a)$.
\end{lemma}

\begin{proof}
  Send the flow on~$e_1$ around~$v$ in the other direction.  
  This \emph{elementary operation} increases the cost for the edges
  $e_2,\dots,e_d$ by at most $\phi(e_1^+)$.  
  The resulting flow~$\phi'$ is feasible and coincides
  with~$\phi$ on all other arcs of~$N_G$.
\end{proof}

\begin{theorem}
  \label{thm:3approx}
  There is an algorithm that computes, for any plane 4-graph~$G$, an
  unbent collection with at most $3\tbn(G)$ bends.  The algorithm
  runs in $O(n^{1 + o(1)})$ time, where $n$ is the number of vertices of~$G$.
\end{theorem}

\begin{proof}
  We first compute a minimum-cost flow~$\phi$ in~$N_G$.  This flow~$\phi$
  corresponds to a bend-minimal orthogonal representation~$R$ of~$G$.
  Based on~$R$, we compute an orthogonal drawing~$\Gamma$ of~$G$~\cite{t-oeggwmnb-sjc}.
  Let~$b$ be the cost of~$\phi$ (which corresponds to the number of bends
  in~$R$ and~$\Gamma$).  If $b=0$, we return $\{\Gamma\}$ as an unbent collection
  (which is optimal), and we are done.

  Otherwise, any optimal unbent collection has size at least~2 and, hence, $\opt \ge 2b$.
We use \cref{prop:arboricity} to 4-color the edges of~$G$
  in linear time such that each color induces a star forest.
We now define an unbent collection of size~4.
  For each $i \in [4]$, we construct a drawing~$\Gamma_i$
  where all edges of color~$i$ are straight.
  Each time, we start with~$R$ and proceed as follows.
  Let $e$ be an edge of color~$i$, and let
  $b_e=\phi(e)$.  If $b_e=0$, we do nothing.  Otherwise we
  straighten~$e$ by \cref{lem:change-flow} with respect to the
  endpoint~$v$ of~$e$ that is {\em not} the center of the star that
  contains~$e$.  As a result, the contribution of~$e$ to the cost
  of~$\phi$ reduces from~$b_e$ to~0 and, for each edge~$e' \ne e$ that
  is incident to~$v$, its contribution to the cost of~$\phi$ increases
  by at most~$b_e$.  Since~$v$ has degree at most~4, the cost of the
  operation is bounded by $-e_b+3b_e=2 b_e$.  Performing all necessary
  elementary operations results
  in a feasible flow~$\phi'_i$ whose cost is at most
  $\cost(\phi)+ \sum_{e \text{ of color } i} 2b_e$.  As the edges of
  color~$i$ form a star forest and we never apply the elementary operation 
  to the center of a star, the operation regarding one edge of
  color~$i$ does not increase the cost of any other edge of color~$i$.
  Thus, using $\phi'_i$, we can compute an orthogonal
  representation~$R_i$ and an orthogonal drawing~$\Gamma_i$ where all
  edges of color~$i$ are straight.
In total, the unbent collection $\{\Gamma_1,\dots,\Gamma_4\}$ has at
  most $4 \cost(\phi) + \sum_{e \in E(G)} 2b_e = 4b+2b = 6b$ bends, and
  the approximation factor of our algorithm is
  $\alg/\opt \le 6b/(2b) = 3$.
The algorithm for computing the minimum-cost flow~\cite{DBLP:conf/focs/Brand0PKLGSS23}
  dominates the running time of our
  algorithm, which hence runs in $O(n^{1+o(1)})$ time.
\end{proof}

\section{The Unbent Bend Number of Plane Triconnected Cubic Graphs} 
\label{se:sub-cubic}

In this section we describe how to compute, for a given plane triconnected cubic graph, an unbent collection with the minimum number of
bends.  The collections that we compute consist of only
two drawings, which is the minimum.

Let $G$ be a plane triconnected cubic graph. Observe that computing an orthogonal representation of a graph with the minimum number of bends is equivalent to subdividing the edges of the graph the minimum number of times so that the graph admits an orthogonal representation without bends. Hence, our strategy
will be to compute two different subdivisions of $G$ such that each of the resulting graphs admits an orthogonal representation without bends and the total number of subdivision vertices is minimum.
We call a subdivision vertex a \emph{dummy}.
Let $f_{\ext}$ denote the external face of $G$. Given a cycle $C$, a \emph{leg} of $C$ is an edge external to $C$ and incident to a vertex of~$C$. 
A \emph{$k$-legged cycle} $C$ of $G$ is a cycle such that: (i)~$C$ has exactly $k$ legs, 
and (ii)~no two legs share a vertex. 
A $k$-legged cycle with $k\le 3$ is \emph{bad} if less than $4-k$ dummies are placed along its edges. 
Otherwise, it is \emph{good}. 
Observe that a 3-legged cycle (such as $C_5$ in \cref{fig:cubic-ortdraw}) is good if we place a dummy along one of its edges, whereas the external cycle is $0$-legged and hence it is good if we place four dummies along its edges (which we did in \cref{fig:cubic-ortdraw}). 
We use the following lemma.

\begin{figure}[b]
  \begin{subfigure}{0.28\textwidth}
    \centering
    \includegraphics[page=1]{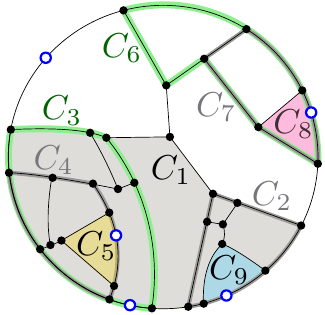}
    \subcaption{$G^1$}
    \label{fig:cubic-ortdraw-a}
  \end{subfigure}
  \hfill
  \begin{subfigure}{0.28\textwidth}
    \centering
    \includegraphics[page=2]{unbent-cubic-main.pdf}
    \subcaption{$G^2$}
    \label{fig:cubic-ortdraw-b}
  \end{subfigure}
  \hfill
  \begin{subfigure}{0.35\textwidth}
    \centering
    \includegraphics[page=3]{unbent-cubic-main.pdf}
    \subcaption{orthogonal drawing of $G^1$}
     \label{fig:cubic-ortdraw-k}
  \end{subfigure}
  \caption{(a--b) Unbent collection $\{G^1,G^2\}$ of a plane triconnected cubic graph~$G$ with 3-legged cycles $C_1,\dots,C_9$ and
    dummies (blue circles).  (c)~Drawing of~$G^1$ without bends.}
  \label{fig:cubic-ortdraw}
\end{figure}

\begin{restatable}{lemma}{cubicbadcycles}\cite{DBLP:journals/jgaa/RahmanNN03}
  \label{app-le:cubic-badcycles}
  A plane cubic graph has an orthogonal representation without bends if and only
  if it does not contain any bad cycle.  Given a plane cubic graph~$G$ without bad
  cycles, an orthogonal representation of~$G$ without bends can be computed in linear time.
\end{restatable}

In \cref{fig:cubic-ortdraw-a,fig:cubic-ortdraw-b} we show two ways to add dummies to a plane triconnected cubic graph~$G$ such that all cycles are good (and such that no edge of~$G$ receives dummies in both cases). 
In the resulting graphs $G^1$ and $G^2$, every 3-legged cycle contains a dummy and the external cycles contain four dummies. 
Due to this, both graphs admit orthogonal drawings without bends; see that 
of~$G^1$ in \cref{fig:cubic-ortdraw-k}.

We say that the $k$ paths of a $k$-legged cycle connecting the vertices incident to its legs are the \emph{contour paths} of the cycle.
A \emph{leg face} of a $k$-legged cycle is a face incident to two of its legs. 
We now aim to lowerbound $\tbn(G)$.
To this end, we use specific 3-legged cycles of~$G$ that we call
\emph{demanding}. Informally, the demanding cycles form the smallest
subset of 3-legged cycles with the property that, if we place a dummy
on an edge of each of them, all 3-legged cycles in~$G$ will be
good. A \emph{non-demanding} cycle is a 3-legged cycle that is not demanding. For a formal definition of demanding cycles, see the
\lncs{full version~\cite{OurArxiv}}\lip{the appendix}.  The set of
demanding cycles is unique.  Let~$D(G)$ be the largest set of
demanding cycles of~$G$ that are pairwise vertex-disjoint.

For an example, see \cref{fig:cubic-ortdraw-a}, where $D(G)=\{C_3,C_5,C_8,C_9\}$. Observe that all these 3-legged cycles are good due to the placement of four dummy vertices, one for each cycle in $D(G)$. 
E.g., the dummy vertex on~$C_2$ also lies on~$C_9$ (since $C_2$ and $C_9$ share an edge). 
Hence, both these 3-legged cycles are good.

The demanding cycles and the set $D(G)$ are computed by using a coloring rule over a structure that describes the hierarchy of the 3-legged cycles of~$G$. This structure and the coloring rule are described in the {\lncs{full version~\cite{OurArxiv}}\lip{the appendix}}. We now lowerbound $\tbn(G)$ in terms of the numbers of certain types of 3-demanding cycles.

A cycle $C$ \emph{contains} another cycle $C'$ if the edges of $C'$ are edges of the subgraph of $G$ induced by~$C$ and its interior vertices.
Let $C$ be a demanding cycle. 
A contour path $P$ of $C$ is \emph{interesting} if every non-demanding cycle of~$G$ that contains~$C$ and shares some edges with~$C$, contains~$P$. 
For example, $C_8$ in \cref{fig:cubic-ortdraw-a} has two interesting contour paths, namely those sharing edges with $C_6$ and $C_7$. 
See also $C_5$ and $C_9$, which have only one interesting contour path each, sharing edges with~$C_4$ and~$C_2$, respectively.

A demanding cycle $C$ is \emph{expensive} if and only if $C$ has exactly one interesting contour path $P$ that contains a single edge, and there exists at least one non-demanding 3-legged cycle $C'$ such that (i)~$C'$ contains $C$, (ii)~$C'$ shares edges with~$P$, and (iii)~$C'$ does not share any edges with other demanding cycles that it contains.
For example, in \cref{fig:cubic-ortdraw-a}, the demanding cycles $C_5$ and $C_9$ are expensive, as both have an interesting contour path composed of one edge and neither of the containing cycles ($C_4$ and~$C_2$) contains other demanding cycles.

A demanding cycle~$C$ is \emph{short} if and only if it has exactly one edge $e$ incident to $f_{\ext}$ and two or more interesting contour paths, one of which is the contour path formed by the edge~$e$.  
For example, in \cref{fig:cubic-ortdraw-a}, cycle $C_8$ is short, as it has two interesting contour paths (shared with $C_6$ and $C_7$), one of them is incident to $f_{\ext}$ and composed by one edge. 

We now define three subsets of $D(G)$ and then use their cardinalities to lowerbound $\tbn(G)$.
Let $D_{\ext}(G)$ be the set of demanding cycles that have $f_{\ext}$ as leg face, let $D_{\exp}(G)$ be the set of expensive demanding cycles, and let $D_{\short}(G)$ be the set of short demanding cycles.

\begin{restatable}[\restateref{app-le:cubic-lower3}]{lemma}{cubiclower}
  \label{app-le:cubic-lower3}
  For any plane triconnected cubic graph $G$, it holds that\\[-2ex]
  \[\tbn(G) \ge q(G) := 2|D(G)|+|D_{\exp}(G)|+8-\min\{8, 2|D_{\ext}(G)|-|D_{\short}(G)|\}.\]
\end{restatable}

In \cref{app-th:minbend-cubic} we prove that the lower bound of \cref{app-le:cubic-lower3} is tight. 
\Cref{fig:cubic-ortdraw} shows an unbent collection $\mathcal{C}=\{G^1,G^2\}$ of plane triconnected cubic graph $G$ with
$D(G)=\{C_3,C_5,C_8.C_9\}$, $D_{\exp}(G)=\{C_5,C_9\}$, $D_{\ext}(G)=\{C_3,C_8,C_9\}$, and $D_{\short}(G)=\{C_8\}$. 
Hence, $\mathcal{C}$ has $\tbn(G) = q(G) = 2 \cdot 4 + 2 + 8 - (2 \cdot 3 - 1) = 13$ dummies. 

\begin{restatable}[\restateref{app-th:minbend-cubic}]{theorem}{minbendcubic}
  \label{app-th:minbend-cubic}
  There exists a linear-time algorithm that computes, for any plane triconnected cubic graph~$G$, an unbent collection of $G$ with $\tbn(G)$ bends. 
  The resulting collection consists of two drawings, which is optimal.
\end{restatable}

\begin{proof}[\sketchname]
We first create two copies of~$G$, which we call~$G^1$ and~$G^2$.  On these, we place $q(G)$ dummies in total. 
Below we prove that, then, none of the two resulting graphs has a bad cycle.  Hence,
by \cref{app-le:cubic-badcycles}, they can be drawn without bends, in linear time.
Replacing the dummies by bends yields an unbent collection with $q(G)$ bends.
Finally, \cref{app-le:cubic-lower3} implies the theorem.

We compute the sets $D(G)$, $D_{\exp}(G)$, $D_{\ext}(G)$, and $D_{\short}(G)$.
We first consider a demanding cycle~$C$ in $D_{\exp}(G)$.
In $G^1$, we place a dummy on the only edge belonging to an interesting contour path. 
In $G^2$, we place a dummy on another edge of $C$ (not belonging to an interesting contour path), so that $C$ is good in both $G^1$ and $G^2$. 
Still in~$G^2$, we place a dummy on one of the two legs of $C$ that is contained in all the non-demanding cycles that contain~$C$ and share edges with~$C$ (in the \lncs{full version~\cite{OurArxiv}}\lip{the appendix} we prove that such a leg always exists). 
Hence, for every cycle in $D_{\exp}(G)$, we have placed three dummies; one in $G^1$ and two in $G^2$. As a result, all demanding cycles in~$D_{\exp}(G)$ and all non-demanding cycles sharing edges with them are good in $G^1$ and $G^2$. 

We now consider the cycles in $D(G)\setminus D_{\exp}(G)$. 
For each such cycle having two or more edges on an interesting contour path, we place a dummy on one such edge in~$G^1$ and one in~$G^2$. 
It remains to place dummies on (i)~demanding (non-expensive) cycles that have one edge belonging to an interesting contour path and (ii)~on the non-demanding cycles that share edges only with the demanding cycles in case~(i).
To this end, we use a data structure that we call \emph{enriched genealogical tree}.
With its help, we place the dummies that we need to place on such demanding cycles exclusively on interesting contour paths. 
As a result, all 3-legged cycles become good. 

According to the placement strategy above, we have placed a total of three dummies for every cycle in $D_{\exp}(G)$ and two dummies for every other demanding cycle.
Hence, we have placed exactly $2|D(G)\setminus D_{\exp}(G)|+3|D_{\exp}(G)|=2|D(G)|+|D_{\exp}(G)|$ dummies. 

Since the external cycle is 0-legged, we have to ensure that it gets four dummies in $G^1$ and four (different ones) in $G^2$.
We show that, for each cycle $C$ in $D_{\ext}(G)\setminus D_{\short}(G)$, we can place two dummies incident to~$f_{\ext}$ such that $C$ and all the non-demanding cycles sharing edges with $C$ become good cycles. 
(For the expensive cycles, we need three dummies, but we have already considered this extra cost before.)
For the cycles in $D_{\short}(G)$, we need three dummies, two of which incident to $f_{\ext}$, to obtain the same results. 
Summing up, on the external cycle, we place $8-\min\{8, 2|D_{\ext}(G)|-|D_{\short}(G)|\}$ dummies that we have not already considered in the term $2|D(G)|+|D_{\exp}(G)|$. 
Hence, in total we have placed  $2|D(G)|+|D_{\exp}(G)|+8-\min\{8, 2|D_{\ext}(G)|-|D_{\short}(G)|\}=q(G)$ dummies. As a result, neither $G^1$ nor $G^2$ has any bad cycle. 
\end{proof}

\section{Open Questions and Further Research}
\label{se:conclusion}

We have introduced the concept of unbent collections.
We propose here several open questions.
Can we efficiently decide, for a given 4-plane graph~$G$,
whether $\un(G) = 2$ or $\un(G) = 3$?
Does every plane 4-graph $G$ admit an unbent collection of size $\un(G)$ that minimizes $\tbn(G)$? 
We have shown how to efficiently compute the unbent bend number of plane triconnected cubic graphs, but there are other natural classes to explore. 
For example, plane cubic graphs that are not triconnected, bipartite planar graphs, and series-parallel 4-graphs are good candidates as they admit unbent collections of size~2. 

Lastly, we discuss a subtlety of \cref{def:unbent}.
Throughout the paper, we have required the drawings in the unbent
collection to have the same embedding as that of the given plane
4-graph.  This requirement stems from the NP-hardness of the bend
minimization problem in the variable-embedding
setting~\cite{DBLP:journals/siamcomp/GargT01}.
However, we suggest for future exploration a variant in which one does not have to adhere to the same embedding in all the drawings in the collection; we call a collection allowing for variable embedding \emph{mutable}.
Additionally, we propose a variant between the two definitions where we still allow only one fixed embedding which is however not given in the input; we call a corresponding collection \emph{uniform}.
Observe that the \emph{mutable unbent bend number} is different from the unbent bend number even considering our introductory example of~$K_4$.
Indeed, the mutable collection can contain two copies of the left drawing in \Cref{fig:k4-b}.
Hence, its mutable unbent bend number is at most 10 (as opposed to $\tbn(K_4)=12$).

\paragraph{Acknowledgments.}  The authors thank the organizers and the other participants of the workshop HOMONOLO 2023 (Nová Louka, Czech Republic), where this work was started.

\lip{\bibliographystyle{plainurl}}
\lncs{\bibliographystyle{splncs04}}
\bibliography{abbrv,unbent}
\lncs{\end{document}}

\appendix

\clearpage

\noindent
{\sffamily\bfseries\Large Appendix}

\section{Omitted Proofs of \cref{se:unbent}}

\Decompose*
\label{lm:decompose*}

\begin{proof}
    First consider a proper induced subgraph $G'$ of $G$ with at least $2|V(G')| - 1$ edges.
    Since $G$ has maximum degree 4 and $G'$ misses at most one edge to be 4-regular, $G'$ is connected to $G \setminus G'$ with at most two edges.
    We use this observation to decompose $G$ recursively as follows.
    If $G$ has a cut $X$ with at most 2 edges, we decompose $G$ at $X$ into two smaller graphs.
    For each component $C$ of $G - X$, we construct the graph $G_C$ by
    taking the union of $C$ and $X$.
    If $|X| = 2$ and the two endpoints of $X$
    that are not in $C$ are the same vertex~$z$,
    we obtain $G_C$ by splitting $z$ into two copies~-- each of degree~$1$.
    Note that in any case, the edges of $X$ in $G_C$ have an
    endpoint of degree~$1$.

    When we recursively continue decomposing each $G_C$
    with respect to a cut~$X$,
    we ignore the edges of~$X$ when searching for a 2-edge cut in all subsequent iterations.
    Since every proper induced subgraph $G'$ of $G$ with at least $2|V(G')| -1$ edges is connected via at most a 2-edge cut to the rest of~$G$, the base case of this recursion consists of graphs that themselves have at least $2n-1$ edges (ignoring the degree-1 vertices) or graphs that have no such high-density subgraphs.
    Because the latter graphs always admit an unbent collection of size at most 2 by Proposition~\ref{prop:DensityNashWilliams}, we can ignore them.

    Since, for every remaining component $C$, the
    corresponding graph~$G_C$ may still have the degree-1 vertices and,
    therefore, does not necessarily have at least ${2 |V(G_C)| - 1}$ edges overall, we attach to each degree-1 vertex a gadget $H$
    (see \cref{fig:Deg1Gadget}) of constant size with $2|V(H)|-1$ edges
    and $\un(G + e) \leq 2$, where $e$ is the edge incident to the degree-1 vertex.
    This ensures that every resulting $n$-vertex 4-graph has at least $2n-1$ edges, which together compose~$\mathcal{G}$.
    
    Regarding the running time,
    note that we can find all 1- and 2-edge cuts simultaneously
    by applying the following procedure.
    Index all faces of $G$ in arbitrary order.
    Traverse each face and label the incident edges by the index of the face.
    This can be done in linear time.
    Now, each edge has a tuple of two labels $(i, j)$.
    Assume that $i \le j$.
    If $i = j$, we have found a 1-edge cut and mark the corresponding edge.
    We sort the tuples of labels lexicographically in $O(n)$ time
    using radix sort.
    In the resulting list, if we encounter two consecutive tuples
    with the same labels, we have found a 2-edge cut and mark the corresponding edges.
    Using a BFS-traversal in linear time,
    we can find all components separated by 1- or 2-edge cuts.

    We now show that $\un(G) \leq 2$ if and only if $\un(G_C) \leq 2$ for every component $C$ of $G - X$.
    If $\un(G) \leq 2$, we can obtain
    an unbent collection of size at most two for each $G_C$
    by taking an unbent collection of size at most two for $G$
    and removing all elements that are not contained in~$G_C$.
    
    Conversely, assume that, for each component $C$,
    there is an unbent collection consisting of
    two drawings $\Gamma_C^1$ and $\Gamma_C^2$.
    For each $i \in [2]$,
    we combine drawings $\Gamma_{C_1}^i, \Gamma_{C_2}^i, \dots$
    by identifying the edges of $X$ in each of these drawings,
    where $C_1, C_2, \dots$ are the components of $G-X$. Of course, it might happen that one of the edges of $X$ needs to be redrawn with some bends after this identification, but this is not a problem because we have two drawings and at most two edges in $X$, so in each drawing, we can choose to make one edge straight and give the other as many bends as we need.
This yields an unbent collection of size two for~$G$.
\end{proof}

\BalancedColoring*
\label{the:balancedColoringCharacterization*}

\begin{proof}
    First assume that $\un(G) \leq 2$.
	Since we can assume that every edge is straight in one of the two drawings, the straight edges induce a 2-edge-coloring of $G$.
	Moreover, for each of the two drawings, the corresponding solution to Tamassia's flow network also induces an angle assignment for all vertices.
	Now assume that there exists a cycle $c$ in one color that encloses faces whose demands do not sum up to 0.
	Since the corresponding edges are straight in the drawing, this means that no flow passes $c$ in the solution to the flow network.
	But since the demands of the faces enclosed by $c$ do not sum up to zero, this implies that the solution to the flow network is invalid; a contradiction.
 
    Conversely, assume that $G$ admits a balanced 2-edge-coloring.
	For each color~$x$, we construct a feasible flow for the flow network $N_G$ that keeps the $x$-colored edges without bends as follows.
	Every vertex passes one unit of flow (i.e., one angle of $90^\circ$)
    to each incident face, the remaining flow produced by this vertex
    is distributed according to the angle assignment.
	Observe that after this, the amount of flow present at the nodes of $N_G$
    representing faces of~$G$ coincides with the demand of these faces as defined above.
    A cycle is $x$-colored if all its edges are $x$-colored. 
    Consider a cycle that is $x$-colored that does not enclose edges of any other cycle that is $x$-colored.
	As the demands of the faces enclosed by $C$ sum up to 0, we can find a feasible flow for the part of $N_G$ corresponding to the inside of $C$ without assigning flow to any edge of $N_G$ that bridges an $x$-colored edge of~$G$.
	Combining these flows for all such cycles and using the same argument for the region outside the outermost cycle, we obtain a feasible flow for $N_G$
    corresponding to an orthogonal representation that does not assign any bends to $x$-colored edges.
Therefore, $\un(G) \le 2$.
\end{proof}

\Counterexamples*
\label{pro:counterexamples*}

\begin{proof}
  By~\cref{Thm:unbentnumberlessthan3} we know that $\un(G)\le 3$.
  Suppose~$\un(G)\le 2$, then
  by~\cref{the:balancedColoringCharacterization} $G$ admits a balanced
  2-edge-coloring.
  
  Condition~(i). Since an acyclic graph
  contains at most $n-1$ edges, one color has a cycle. Since that
  cycle contains only triangles on the inside, it cannot be balanced
  (note that lower degree vertices do not help since they provide
  supply and not demand like triangles). Hence, $\un(G)=3$.

  Condition~(ii). Note that $G$ is 4-regular.
  Assume one color has two cycles. Then the two cycles are both
  balanced which is only possible, if they both contain $f$ on the
  inside. But then the symmetric difference of their insides is
  enclosed by one color while only containing triangles. Since $G$ is
  4-regular, this cannot be balanced. Hence, in each color there is at
  most one cycle. Therefore, each color can have at most $n$ edges and
  thus each color is spanning with exactly one cycle (and each edge is
  only assigned one color).
   
  Now consider the edges of $f$. If $f$ shares an edge $e$ with the
  external face, then the cycle of the color not used by $e$ does not
  contain $f$ and can thus not be balanced. Hence, each edge of $f$ is
  incident to a triangle since all inner faces except $f$ are
  triangles.  Hence, for each edge $e$ of $f$ that is not assigned to
  a color $c$, the cycle of $c$ contains the triangle incident to $e$.
   
  Assume that such a triangle $t$ shares two edges $e_1$, $e_2$ with
  $f$. Then $e_1,e_2$ share a vertex $v$ since a triangle only has
  three vertices and two non-adjacent edges have 4 in total. But then
  $v$ cannot have any other incident edges since such an edge would
  lead inside $f$ or $t$, which would contradict $f$ and $t$ being
  faces of a plane graph.
However, $v$ has four incident edges since $G$ is
  $4$-regular.  Hence, each edge of $f$ is incident to a distinct
  triangle, and, for each edge of $f$, one of the two monochromatic
  cycles contains another triangle.  By our analysis of face demands
  (and ultimately Tamassia's flow network~\cite{t-oeggwmnb-sjc}),
  a balanced cycle containing~$f$ also contains exactly $k-4$ triangles,
  where $k$ is the number of edges of $f$. Consequently, both
  monochromatic cycles together should contain exactly $2k-8$ triangles,
  while we know that they contain at least $k$.
  Since $k\le7$ and, then, $2k - 8 < k$, this is a contradiction.
  Hence, $\un(G)=3$.
\end{proof}

\section{Omitted Proofs of \cref{se:bends}}

\subsection{NP-Completeness of Determining the Unbent Bend Number}

Here, we provide the full proof that determining the unbent bend number
of a 4-plane graph is NP-complete.

We start by proving that the problem \textsc{LevelPlanarMonotone2in4Sat} is NP-complete.
In \textsc{2in4Sat}, the input is a Boolean formula in conjective normal form where every clause consists of four literals, 
and the task is to decide whether an assignment of variables to Boolean values
exists such that in every clause exactly two literals are true.
In the problem \textsc{LevelPlanarMonotone2in4Sat}, additionally, the variable--clause incidence graph is planar and admits a drawing where the variable vertices are
placed on a set of parallel lines, called \emph{levels}, and the clause vertices are
placed between the levels such that no edge crosses a level.
The levels divide the plane into rectangular strips.
The clauses that are represented in the first, third, etc.\ strip contain only negated literals while
the clauses represented in the second, fourth, etc.\ strip contain only unnegated literals.

\LevelTwoInFourSAT*
\label{lm:level2in4satNPh*}

\begin{proof}
    The membership in NP is obvious.  To show the NP-hardness,
    we reduce from the NP-hard problem \textsc{PlanarPositive2in4Sat}~\cite{kkw-cbvop-07},
    which is \textsc{2in4Sat} with the additional property that the incidence graph is planar and all literals are unnegated.
    As a first step, we reduce to \textsc{LevelPlanarPositive2in4Sat},
    where we require a leveling as in \textsc{LevelPlanarMonotone2in4Sat}
    but all literals are unnegated.
    Consider a planar straight-line drawing of the incidence graph,
    which exists and can be computed in polynomial time~\cite{fpp-hdpgg-90,DBLP:conf/soda/Schnyder90}.
    Rotate the drawing slightly such that no pair of a variable vertex
    and a clause vertex lies on the same y-coordinate.
    Take the set of horizontal lines going through
    the variable vertices as levels.\footnote{In the resulting leveling, every
        variable vertex may lie on its own level.}
    There may be edges that cross levels.
    We now show how to reduce the number of edge--level crossings
    by one in each iteration.
For an illustration, see \cref{fig:make-level-planar}.
    
    \begin{figure}[t]
        \begin{subfigure}[t]{.4 \textwidth}
            \centering
            \includegraphics[page=1]{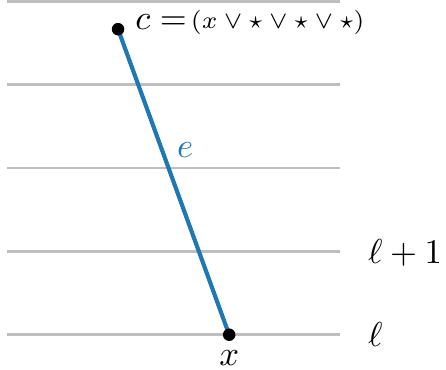}
            \caption{Long edge~$e$ intersects at least one level.}
            \label{fig:make-level-planar-a}
        \end{subfigure}
        \hfill
        \begin{subfigure}[t]{.5 \textwidth}
            \centering
            \includegraphics[page=2]{levelplanarsat}
            \caption{An equivalent instance where one edge--level
                intersection has been removed.}
            \label{fig:make-level-planar-b}
        \end{subfigure}
        \caption{Making a non-level planar incidence graph iteratively level planar.}
        \label{fig:make-level-planar}
    \end{figure}
    
    Let $e = xc$ be a ``long'' edge between a variable vertex~$x$ and a clause vertex~$c$.
    Without loss of generality, assume that $x$ lies
    on some level~$\ell$ and $c$ lies above level $\ell + 1$.
    Introduce two new variables $y$ and $z$
    and place the corresponding vertices on level~$\ell + 1$
    directly to the left and to the right, respectively, of the intersection point
    of~$e$ and level~$\ell + 1$.
    In the clause $c$, we replace every occurrence of $x$ by $z$,
    which we can do if we guarantee that $x = z$.
    Note that in the incidence graph,
    this means that we replace the edge~$e$ by a new edge~$e' = zc$.
    We can do this in a planar way because, starting in~$c$, we can route $e'$ such 
    that it follows the course of~$e$ until it is in the vicinity of~$z$.
    Next, we ensure that, in any \textsc{2in4Sat} solution,
    we have $x = z \ne y$.
    To this end, we add the clauses $(x \lor x \lor y \lor y)$
    and $(y \lor y \lor z \lor z)$.
    Note that, in the drawing of the incidence graph,
    the two corresponding clause vertices can be added
    directly below $y$ and $z$ without creating crossings.
    
    This iterative process yields an equivalent instance
    of \textsc{LevelPlanarPositive2in4Sat} where no edge crosses a level.
    Note that in any \textsc{2in4Sat} instance, we can flip
    all four literals in an arbitrary clause and get an equivalent instance.
    We do this for all clauses whose vertices lie in the
    first, third, etc.\ strip formed by pairs of consecutive levels.
    This yields an equivalent instance of \textsc{LevelPlanarMonotone2in4Sat}.
\end{proof}

\NPh*
\label{th:nph*}

\begin{proof}
    Containment in NP is clear.
    We show NP-hardness by reduction from the problem \textsc{LevelPlanarMonotone2in4Sat}.
    Let $\Phi$ be the given \textsc{2in4Sat} formula,
    let $m$ be the number of clauses in $\Phi$,
    let $n$ be the number of variables in $\Phi$,
    let $H$ be the given embedded (level-planar)
variable--clause incidence graph of $\Phi$,
    and let $\ell$ be the number of levels of~$H$.
    We now construct a plane 4-graph~$G$
    that follows the structure of the incidence graph~$H$.
    The general structure repeats at every level.
    \Cref{fig:nph-construction-app} illustrates the construction
    for one level of variable vertices; the non-negated clauses
    are above that level and the negated clauses are below.
    
    \begin{figure}[tbh]
        \centering
        \includegraphics[page=1]{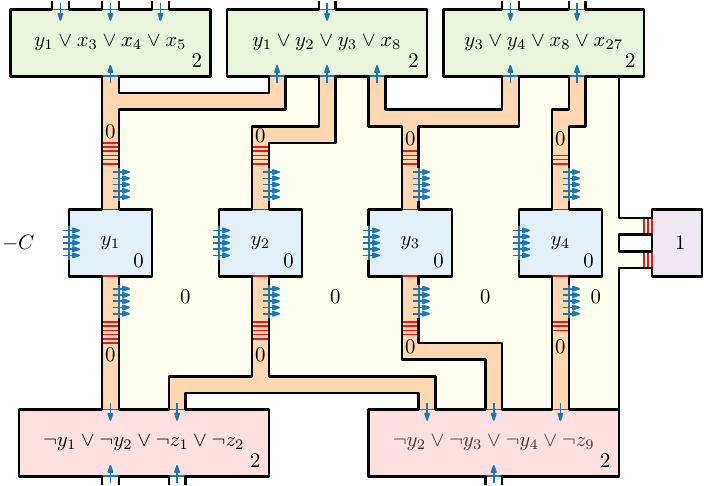}
        \caption{NP-hardness reduction: A layer of variable gadgets and the neighboring clause gadgets.
        Numbers inside faces (including the external face) indicate the demands of the faces.
        This figure is a copy of \cref{fig:nph-construction} for the reader's convenience.}
        \label{fig:nph-construction-app}
    \end{figure}
    
    \paragraph{Overview.}
    In $G$, there is a face for each variable vertex (\emph{variable gadget};
    see \cref{fig:nph-variable}),
    for each clause vertex (\emph{clause gadget};
    see \cref{fig:nph-clause}),
    and for each set of edges that are incident and above (below)
    the same vertex (\emph{tunnel gadget}).
    The gadgets are attached to each other according to~$H$.
    The gadgets are designed such that
    they need to receive specific (non-negative) flow
    values in Tamassia's flow model.
    They are surrounded by a set of faces with demand~0 called \emph{walls},
    which make it costly for flow to enter the gadgets
    unless it is through some \emph{door}.
    (All walls have ``thickness'' $2n$, that is,
    each unit of flow traversing a wall bends at least $2n$ edges.)
    We construct each face with a specific demand.
    Most faces have demand 0 (so they can be drawn without bends as they are),
    the clause gadgets have demand~2, and an extra face in each level
    has demand~1 (called \emph{level gadget}, shown on the right
    in \cref{fig:nph-construction-app}).
    These demands sum up to $C = 2 m + \ell$.
    The external face is the only face with a negative demand (namely, $-C$).
    In other words, the external face has supply~$C$
    and no other face supplies flow.
    We will guarantee that there is an unbent collection (with two drawings)
    and at most $k = 2 n C$ bends if and only if
    $\Phi$ is a yes-instance.
    The key property is that the flow needs to avoid walls,
    which are, in terms of bends, too expensive to traverse,
    and hence the flow must traverse the variable and tunnel gadgets to
    enter the clause gadgets and the level gadgets.
    Since every variable gadget has one door to enter and
    two exit doors to leave the gadget (the latter being single edges),
    in the one drawing it must take the one exit door and in the other
    drawing it must take the other exit door.
    This choice corresponds to an assignment of Boolean values to variables.
    
    We next describe the gadgets of the plane graph~$G$ in more detail.
    
    \begin{figure}[t]
        \begin{subfigure}{.16 \textwidth}
            \centering
            \includegraphics[page=2]{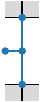}
            \caption{base state}
            \label{fig:valve:base}
        \end{subfigure}
        \hfill
        \begin{subfigure}{.18 \textwidth}
            \centering
            \includegraphics[page=3]{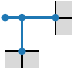}
            \caption{1-bent state}
            \label{fig:valve:bent}
        \end{subfigure}
        \hfill\hfill
        \begin{subfigure}{.16 \textwidth}
            \centering
            \includegraphics[page=4]{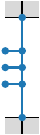}
            \caption{base state}
            \label{fig:valve:base3}
        \end{subfigure}
        \hfill
        \begin{subfigure}{.2 \textwidth}
            \centering
            \includegraphics[page=5]{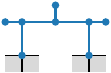}
            \caption{2-bent state}
            \label{fig:valve:bent3}
        \end{subfigure}
        \caption{(a,b): a 1-valve gadget, (c,d): a 3-valve gadget.  A $k$-valve gadget allows us to distribute $k$ integral units of flow from the face to its left to the face on its right (or vice versa), without bending any edge.}
        \label{fig:valve}
    \end{figure}

    \paragraph{Valve gadget.}
    A \emph{$k$-valve gadget} is a plane subgraph of $G$ that
    is a path of $k+2$ vertices where each of the
    $k$ inner vertices gets a leaf attached.
    Along the path, all leaves are
    attached on the same side.
    The path separates two faces~-- we call
    the one containing the degree-1 vertices \emph{source face}
    and the other face \emph{target face}.
    Since each inner vertex of a $k$-valve gadget has degree~3,
    it has one free angle, which can be assigned
    either to the source or to the target face.
    In the \emph{base state}, the target face receives all free angles of the inner vertices
    (see \cref{fig:valve:base,fig:valve:base3}), whereas
    in the \emph{$t$-bent state}, the source face receives
    $t$ free angles (and the target face receives
    $k-t$ free angles)
    from the inner vertices where $t \in [k]$
    (see \cref{fig:valve:bent,fig:valve:bent3}).
If $k$ is clear from the context or irrelevant,
    we sometimes call a $k$-valve gadget just valve gadget.
    For the sake of presentation, we assume that all
    valve gadgets are by default in the base state.
    
    All vertices in $G$
    (except for the inner vertices of the valve gadgets)
    have either degree~1 or~4.
    This means that there is no flexibility in where
    they assign their free angle:
    each degree-1 vertex assigns its three free angles to its only incident face
    and the degree-4 vertices do not have any free angles.
    Hence, we know for every face of~$G$ its precise demand
    (assuming all valve gadgets are in the base state).
    The demands of all faces sum up to zero\footnote{This is a property of all 4-plane graphs
        with respect to Tamassia's flow network}~--
    hence, in a proper flow network, there is
    an assignment of flow that fulfills the demands.
    This flow can be assigned to edges of~$G$,
    which corresponds to these edges being bent in an orthogonal drawing of~$G$.
    Now, when a $k$-valve gadget is in the $t$-bent state,
    this can be seen as $t$ units of flow
    being transferred from the source to the target face.
    Unlike flow going over edges, this does not
    create bends, however, at most $k$ units
    of flow can be transferred in this way
    at a $k$-valve gadget.
    
    \begin{figure}[t]
        \centering
        \begin{subfigure}[t]{.468 \textwidth}
            \centering
            \includegraphics[page=8]{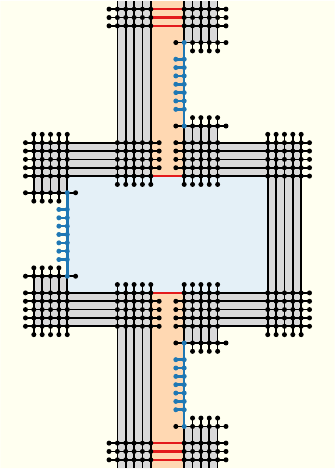}
            \caption{}
            \label{fig:nph-variable-base}
        \end{subfigure}
        \hfill
        \begin{subfigure}[t]{.522 \textwidth}
            \centering
            \includegraphics[page=11]{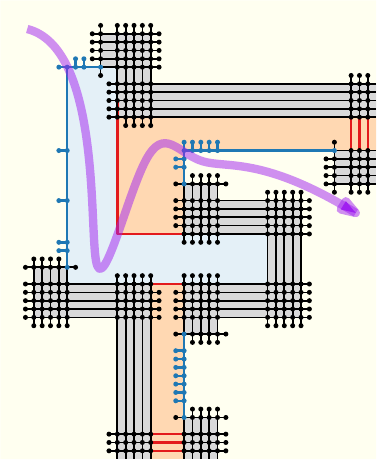}
            \caption{}
            \label{fig:nph-variable-flowed}
        \end{subfigure}
        \caption{Variable gadget. In (a), no flow enters or leaves,
            that is, all valve gadgets are in the base state
            and all (red) exit doors and the inner (red) edges
            of the tunnel gadgets are unbent.
            In (b), one unit of flow (purple arrow)
            enters on the left, leaves via the (red) top exit door
            and leaves the tunnel gadget
            through the valve gadget on the top right side.
            No flow leads to any of the adjacent clause gadgets.}
        \label{fig:nph-variable}
    \end{figure}
    
    \paragraph{Variable gadget.}
    There are $n$ variable gadgets~-- one per variable.
    For an illustration of a variable gadget together
    with the attachment to the tunnel gadgets above and below, see
    \cref{fig:nph-variable}.
    A variable gadget has a central face (light blue background) that
    is bounded by walls except for an $(2m+1)$-valve gadget (shown in
    dark blue on the left)
    and two \emph{exit doors}, which are two single edges
    (shown in red incident to the light blue area).
    All faces have demand 0, however, to allow flow to connect
    the external face with the clause gadgets and the level gadgets
    without going through walls (which is too expensive in terms of bends),
    some flow must enter, and the same flow value must leave each variable gadget.
    The arrangement of variable gadgets (see \cref{fig:nph-construction-app})
    requires that flow from the external face or from the variable
    gadget to the left enters at the $(2m+1)$-valve gadget and leaves
    at the exit doors.
    Note that at most one exit door can be bent per drawing
    and, hence, all flow traversing the variable gadget leaves
    trough the same exit door in a drawing.
    On the other hand, the $(2m+1)$-valve gadget can be used
    in all drawings of the unbent collection as long
    as at most $t \le 2m + 1$ units of flow are involved
    (then it is in the $t$-bent state instead of the base state
    and does not create bends on edges).
    
    \paragraph{Tunnel gadget.}
    There are $2n$ tunnel gadgets~-- two per variable.
    For an illustration of tunnel gadgets
    (shown as slim tunnels with orange background),
    see \cref{fig:nph-construction-app}.
    For the connection to a variable gadget, see
    \cref{fig:nph-variable},
    and, for the connection to a clause gadget, see
    \cref{fig:nph-clause}.
    Each tunnel gadget connects an exit door of a variable gadget
    with all negated or with all unnegated clause gadgets
    such that this variable occurs in the corresponding clause.
    Like all of our gadgets, it is bounded by walls.
    Flow that enters the tunnel gadget can leave it immediately
    via a $(2m+1)$-valve gadget
    to a face leading to the next variable gadget or to a level gadget
    (described below; see the large yellow faces
in \cref{fig:nph-construction-app}).
    This way flow from the external face can reach all variable gadgets.
    The path to the clause gadget is intermitted
    by a set of edges (depicted as the red horizontal line segments
    in \cref{fig:nph-construction-app,fig:nph-variable}
    that are not the exit doors of variable gadgets).
    The number of these edges depends on which variable gadget is affected:
    a tunnel gadget of the $i$-th leftmost variable gadget
    on a level has $n-i$ such edges.
    In this way we ensure that every unit of flow reaching a clause gadget
    from the external face causes at least $n$ bends:
    $i$ bends for traversing $i$ variable gadgets and
    $n - i$ bends for traversing a tunnel gadget of
    the $i$-th variable gadget of some level.

    \paragraph{Level gadget.}
    There are $\ell$ level gadgets~-- one per level.
    For an illustration of a level gadget,
    see the rightmost (purple) face
    with a ``1'' in \cref{fig:nph-construction-app}.
    It is bounded by walls except for
    two doors with $n - s(j)$ edges each
    (shown in red on the left side of the level gadget)
    where $s(j)$ is the number of variables assigned to level~$j$.
    These doors connect the level gadget with the face
    following the $s(j)$-th variable gadget.
    Since it has a demand of~1,
    any flow coming from the external face effects
    at least $n$ bends~-- by going first through
    all variable gadgets of the level and
    then entering the level gadget via one of its two doors.
    The purpose of the level gadget is to
    bend an exit door of each variable gadget of this level
    if traversing walls is avoided.
    Otherwise, if, say, the rightmost variable gadget does
    not forward any flow to its adjacent clause gadgets
    in the first drawing, then it would not bend any
    exit door in the first drawing and could bend both
    exit doors in the second drawing.

    \begin{figure}[t]
        \centering
        \includegraphics[page=14]{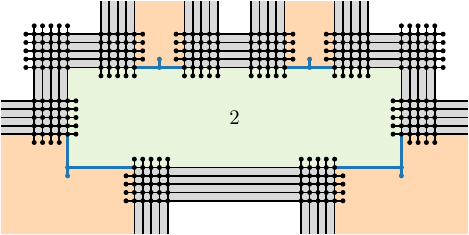}
        \caption{Clause gadget, whose 1-valve gadgets leading to the
            two variable gadgets above are in the base state,
            whereas the valve gadgets leading to the two
            variable gadgets below are in the 1-bent state.
            This corresponds to a clause that is twice satisfied.} \label{fig:nph-clause}
    \end{figure}
    
    \paragraph{Clause gadget.}
    There are $m$ clause gadgets~-- one per clause.
    For an illustration of a clause gadget together
    with the attachment to the tunnel gadgets, see
    \cref{fig:nph-clause}.
    A clause gadget is bounded by walls
    but has four doors connecting it to four
    tunnel gadgets via a 1-valve gadget.
    These four tunnel gadgets lead to the four
    corresponding variable gadgets of the clause.
    A clause gadget has a demand of~2,
    which can be satisfied without creating any
    bends if and only if exactly two of the 1-valve gadgets
    are in the 1-bent state.
    
    This completes the description of the 4-plane graph~$G$.
    Clearly, the size of $G$ is polynomial in~$\Phi$:
    the number of variable, clause, tunnel, and level gadgets
    is in $O(n + m)$ and each gadget including its walls
    has $O(n)$ vertices and edges.
    Thus, $G$ has as $O(n^2 + nm)$ edges and vertices
    and can be constructed in polynomial time.
    We are now ready to prove the correctness of our reduction:
    There is a satisfying truth assignment for $\Phi$
    if and only if $G$ admits an unbent collection
    with at most $k = 2 n C$ bends.
    
    \paragraph{Correctness.}    
    We first assume that there is a satisfying truth assignment for~$\Phi$.
    We show that then, there exists an unbent collection of size~2
    such that each drawing has $n C$ bends.
    Consider the flow network associated with~$G$.
    We specify two flow assignments, which in turn specify two
    orthogonal representations
of an unbent collection.
    For each clause, there are two paths from the
    external face, via each variable gadget whose corresponding variable
    satisfies this clause, to the corresponding clause gadget.
    Each such path causes $n$ bends as described before.
    For each level, there is a path from the external face,
    via all variable gadgets of the level, to the level gadget.
    Each such path causes $n$ bends as described before.
    Since every variable gadget delivers at most one unit
    of flow to a clause gadget, the 1-valve gadgets at the
    clause gadgets suffice
    to let the flow enter the clause gadgets
    without extra bends.
    Similarly, at most $2m + 1$ units of flow are needed per level,
    so the $(2m+1)$-valve gadgets at the variable and tunnel gadgets suffice.
    In the second drawing, we take the remaining paths
    from the external face to the clause gadgets~-- each one
    traversing the variable gadget whose corresponding variable
    does not satisfy the clause.
    Since $\Phi$ is a \textsc{2in4Sat} instance,
    there are precisely two such variables per clause.
    Clearly, this satisfies all demands of the faces.
    In each of the two drawings, we have $n C$ bends, where $C = 2 m + \ell$.
    It remains to argue that every edge is bent in at most one
    drawing of the unbent collection.
    In every variable gadget, either the top or bottom exit door
    bends depending on whether the corresponding variable
    is assigned to true or to false.
    Then, only inner edges of the tunnel gadget
    connecting to this exit door of the variable gadget are bent.
For the level gadget, there are two doors~-- one for each of the two drawings.
    No other edges, and in particular no edges of valve gadgets, are bent.
    
    Now, we assume that there is an unbent collection with at most
    $2 n C$ bends.  Clearly, there is no bend-free drawing, so the
    unbent collection has size at least two.  For each of the
    drawings, there is a corresponding flow network that satisfies the
    demands of the clause and the level gadgets by delivering all flow
    from the external face.  By construction, there is no path between
    the external face and any clause or level gadget that traverses
    less than $n$ edges in the flow network.  Such a path traverses a
    sequence of variable gadgets before it traverses a tunnel gadget
    or enters a level gadget.  Since the total demand of the inner
    faces is $C = 2 m + \ell$, each drawing has at least $n C$ bends.
    Since the total number of bends is $2 n C$, there are precisely
    two drawings with $n C$ bends each, and these bends lie along the
    paths described before.  Note that no edge is bent in both
    drawings, hence, the top and bottom exit door of each variable
    gadget is bent in exactly one of the two drawings~-- this yields
    a variable assignment: if, in the first drawing, the exit door
    of a variable gadget towards the clause gadgets
    corresponding to unnegated clauses is bent, we set the variable
    to true; otherwise, we set it to false.  At the clause
    gadgets, no edge is bent (since this would give more than~$n$
    bends along a path from the external face to the clause gadget)
    and, hence, in each of the two drawings exactly two of the 1-valve
    gadgets are in the 1-bent state.  Only those 1-valve gadgets that
    connect via a tunnel gadget to a variable gadget that is set to
    true (in the first drawing) or to a variable gadget that is set to
    false (in the second drawing) can be in the 1-bent state.
    As, in each of the two drawings,
    we have two such variable gadgets per clause gadget,
    the drawing of the variable gadgets corresponds
    to a satisfying truth assignment for~$\Phi$.
\end{proof}

\subsection{Star Arboricity of Bounded-Degree Graphs}

\Arboricity*
\label{prop:arboricity*}

\begin{proof}
    We express the assignment of the edges of $G$ to
    star forests via an edge coloring using at most $\Delta$ colors. 
    We color the edges greedily by processing the vertices one by one.
    For each vertex $v$, we go through its incident edges that have not yet been colored.  
    For each such edge, we use a color that has not been used for edges incident to~$v$ before. 
    Since there are $\Delta$ colors and $v$ has at most $\Delta$ incident edges, our coloring is such that only pairs of edges from~$v$ to vertices processed before~$v$ are assigned the same color. 
    Note that, if two adjacent edges $xy$ and $yz$ are assigned the same color $c$ in our coloring, then $y$ is the vertex that has been processed last among $\{x,y,z\}$. 
    This means that $x$ and $z$ cannot be incident to other edges of color~$c$. 
    For other edges of color $c$ incident to $y$, the same reasoning applies. 
    Thus, $xy$ and $yz$ are indeed part of a star in color $c$ centered at $y$.
    Hence, every color induces a star forest.
\end{proof}

\section{Ommitted Proofs of \cref{se:sub-cubic}}
\label{se-ap:sub-cubic}

In this section we describe how to compute an unbent collection having the total minimum number of bends. Our collection  consists of two drawings, which is the minimum possible number. 

Let $G$ be a plane triconnected cubic graph with $n$ vertices. Observe that computing an orthogonal representation of a graph with the minimum number of bends is equivalent to subdividing the edges of the graph the minimum number of times so that the graph admits an orthogonal representation without bends. Hence, our strategy will be to subdivide the edges of $G$ the minimum number of times so that the subdivided graph admits an orthogonal representation with no bend. We shall often call \emph{dummy} a subdivision vertex and we shall often use the expression ``place $h$ dummies along and edge'' to mean that the edges has been subdivided $h$ times. Also, in the following $f_{\ext}$ denotes the external face of $G$.

\subsection{Bad Cycles and Demanding Cycles}
\label{app-se:cubic-bad-cycles}
Given a cycle $C$, a \emph{leg} of $C$ is an edge external to $C$ and incident to a vertex of $C$. A \emph{$k$-legged cycle} $C$ of $G$ is a cycle such that: (i) $C$ has exactly $k$ legs (i.e., a $k$ legged cycle is not a $k'$-legged cycle if $k\not =k'$), and (ii)~no two legs share a vertex. Observe that, since $G$ is cubic, there is no pair of legs that can be incident to a same vertex along $C$.
Note that the external cycle of $G$ is 0-legged. The cycles having exactly $k$ legs such that these legs are all incident to the same vertex are called \emph{degenerate $k$-legged cycles}~\cite{DBLP:conf/soda/DidimoLOP20}. Observe that a degenerate $k$-legged cycle is not a  $k$-legged cycle, as (ii)~does not hold. 

A \emph{bad cycle} is a $k$-legged cycle where $k\le 3$ and such that less than $4-k$ dummies are placed along its edges. If a cycle is not a bad cycle, it is a \emph{good cycle}. For example: a 3-legged cycle is bad if no dummy is placed along an edge of it; the external face, which is a $0$-legged cycle, is a bad cycle if less than $4$ dummies are placed along its edges. The following lemma gives a sufficient and necessary condition for a graph to admit an orthogonal drawing without bends and also states that, in this case, there exists a liner-time algorithm computing an orthogonal representation of such graph with no bend.

\cubicbadcycles*
\label{app-le:cubic-badcycles*}

Let $e=uv$ be an edge of the external face~$f_{\ext}$ of $G$ and let $C$ be the cycle forming the external boundary of $G \setminus e$. Observe that, after placing one or more dummies along $e$, cycle $C$ becomes a (possibly degenerate) 2-legged cycle. See for example  \cref{app-fig:2leggedbydummies}. The following lemma gives a sufficient condition for such 2-legged cycles and for degenerate cycles to not exist.

\begin{figure}[t]
\centering
\includegraphics[]{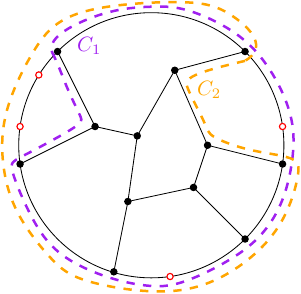}
  \caption{Placing dummy vertices makes cycles $C_1$ and $C_2$ 2-legged cycles. Cycle $C_2$ is degenerate. Both $C_1$ and $C_2$ are good cycles.}
  \label{app-fig:2leggedbydummies}
\end{figure}

\begin{lemma}[\cite{DBLP:conf/soda/DidimoLOP20}]
\label{app-le:2legged}
If we place  dummies along three distinct edges of $f_{\ext}$ or along two independent edges of $f_{\ext}$ so that no edge contains three dummies,
then the resulting graph has no bad degenerate cycles and no bad 2-legged cycles. 
\end{lemma}

See for example \cref{app-fig:2leggedbydummies}, where the conditions of \cref{app-le:2legged} hold and $C_1$ and $C_2$ are good cycles, as they both have two dummies along their edges. Thanks to \cref{app-le:2legged}, we can focus our attention mainly on 3-legged cycles and observe that the other bad cycles (degenerate or 2-legged) are satisfied by the way we place dummies.

Let $C$ be a 3-legged cycle of $G$. We say that the three paths of $C$ connecting the vertices of the cycle incident to the legs are the \emph{contour paths} of $C$. A face incident to a leg of a 3-legged cycle is called \emph{leg face}. Note that a 3-legged cycle has exactly three leg faces. Cycles $C_1,...,C_8$ of  \cref{app-fig:cubic-cycles-a} are all 3-legged cycles. Let $G_C$ be the subgraph of $G$ induced by $C$ and its interior vertices and let $n_C$ the number of vertices of $G_C$. We say that a cycle $C$ \emph{contains} another cycle $C'$ if $C$ is a subgraph of $G_C$. The \emph{genealogical tree} $T(C)$ of a 3-legged cycle $C$ is a tree such that: The root of the tree is $C$; the children of $C$ are the 3-legged cycles of $G_C$ such there is no other cycle $C'\subset G_C$  containing them. \cref{app-fig:cubic-cycles-a} shows $T(C_1)$. Two 3-legged cycles \emph{intersect} if they have an edge in common, but there is no containment relationship between them.

Perform a bottom-up traversal of a genealogical tree $T(C)$ and let $C'$ be the currently visited 3-legged cycle. A \emph{red-green coloring} of the three contour paths of $C'$ is defined by applying the following rules. 

\smallskip\noindent
{\sffamily\bfseries Rule~1:} If no contour path of $C'$ includes a green contour path of one of its children as a sub-path, color green all contour paths of $C'$.

\noindent
{\sffamily\bfseries Rule~2:} Otherwise, color green every contour path of $C'$ that
includes a green contour path of one of its children as a sub-path; color red any other contour path of $C'$.

\smallskip

Cycle $C'$ is \emph{demanding} if its red-green coloring obeys Rule~1 and \emph{non-demanding} otherwise. For an example of application of the coloring rule, see \cref{app-fig:cubic-cycles-a}.

Let $C$ be a cycle incident to the external face $f_{\ext}$, i.e., sharing at least an edge with $f_{\ext}$. The \emph{twin} of $C$, denoted by~$C^t$, is the 3-legged cycle having the same legs of $G$. Notice that also $C^t$ is a cycle incident to $f_{\ext}$. Let $I(G)$ be the set of demanding cycles of $G$ that are intersecting.

\begin{lemma}[\cite{DBLP:conf/soda/DidimoLOP20}]
\label{app-le:cubic-only_intersecting}
Let $G$ be such that $I(G)\not= \emptyset$, let $C\in I(G)$, let $C^t$ be its twin,  let $e_1$ be any edge of $C$, and $e_2$ be any edge of $C^t$. Then: (1)~$C\in I(G)$ is incident to $f_{\ext}$; (2)~if both $e_1$ and $e_2$ are edges of $f_{\ext}$, any other cycle of $I(C)$ passes through either $e_1$ or $e_2$; (3)~$D_{\ext}(G)=\emptyset$ and vice versa, i.e., if $D_{\ext}(G)\not =\emptyset$ then we have $I(G)= \emptyset$.
\end{lemma}

We use the following notations: $b(G)$ denotes the minimum number of dummies that $G$ needs to have no bad cyle;  $D(G)$ is the set of demanding cycles of $G$ that are pairwise non-intersecting; $D_{\ext}(G)$ is the subset of $D(G)$ whose cycles are incident to the external face $f_{\ext}$ (i.e., such that they have $f_{\ext}$ as one of their leg faces). 
Note that, by definition (see Rule~1), the demanding cycles are edge disjoint; also, recall that the external cycle is a 0-legged cycle and consequently it requires four dummies to not be a bad cycle. It follows that it is necessary to place a dummy in every demanding cycle and four dummies in the external cycle in order to have no bad cycle. 
The following lemma gives the value of $b(G)$ as a function of $D(G)$ and $D_{\ext}(G)$.

\begin{lemma}[\cite{DBLP:conf/soda/DidimoLOP20,DBLP:journals/jgaa/RahmanNN03}]
\label{app-le:cubic-fixemb-bends}
$b(G)= |D(G)|+4-\min\{4, |D_{\ext}(G)|\}$.
\end{lemma}

\begin{figure}[t]
  \begin{subfigure}{0.48\textwidth}
    \centering
    \includegraphics[width=0.65\textwidth,page=19]{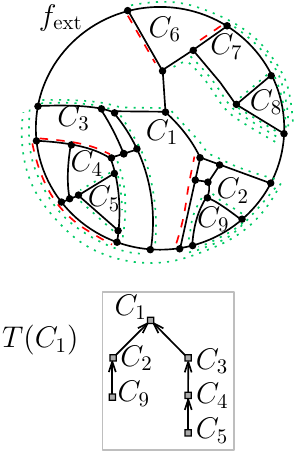}
    \subcaption{}
     \label{app-fig:cubic-cycles-a}
  \end{subfigure}
  \begin{subfigure}{0.48\textwidth}
    \centering
    \includegraphics[page=2]{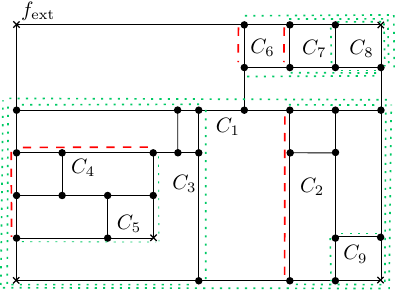}
    \subcaption{}
    \label{app-fig:cubic-cycles-b}
  \end{subfigure}
  \hfill
  \caption{(a) Graph $G$, its 3-legged cycles, and the genealogical tree $T(C_1)$. We have $D(G)=\{C_3,C_5,C_8,C_9\}$, $D_{\ext}(G)=\{C_3,C_8,C_9\}$. (b) A drawing of $G$ with no bends and $b(G)$ dummies, where $b(G)=|D(G)|+4-\min\{4, |D_{\ext}(G)|\}=5$.}
  \label{app-fig:cubic-cycles}
\end{figure}

See for example \cref{app-fig:cubic-cycles-b}, which depicts a orthogonal drawing with no bends and with minimum dummies of the graph in  \cref{app-fig:cubic-cycles-a}, where dummies are represented with a cross. If dummies are removed with bends the drawing has the minimum number of bends.

\subsection{A Better Lower Bound for tbn(G)}
\label{app-subse:cubic-lower}

In this section we refine a lower bound for $tbn(G)$ proving \cref{app-le:cubic-lower3}. In \cref{app-subse:cubic-alg} we prove that such lower bound is tight. Before proving \cref{app-le:cubic-lower3} we prove and observe intermediate results. 

Observe that, for each unbent collection $\mathcal{C}$ of $G$, $|\mathcal{C}|\ge 2$, as the external face of $G$ is a 0-legged cycle, and by \cref{app-le:cubic-badcycles}, we have to place at least four dummies in $G$. We introduce the following corollary of \cref{app-le:cubic-fixemb-bends}.

\begin{corollary}
\label{app-co:cubic-lower1}
$\tbn(G)\ge 2|D(G)|+8-\min\{8, \ 2|D_{\ext}(G)|\}$
\end{corollary}

Before introducing the next intermediate result, which is \cref{app-le:cubic-lower2}, we need to define more concepts.

\smallskip 
Let $C$ be a demanding cycle and let $P$ be a contour path of $C$. Path $P$ is \emph{interesting} if it is a sub-path of some contour paths in all non-demanding cycles of $G$ containing $C$ and sharing some edges with $C$. See \cref{app-fig:cubic-interesting-a} and consider $C$, whose legs are $\{e_1,e_2,e_3\}$; $C'$, whose legs are $\{e_1,e_4,e_5\}$, and $C''$, whose legs are $\{e_6,e_7,e_8\}$. Cycle $C''$ contains $C'$, and $C'$ contains $C$. Cycle $C'$ is demanding and its contour path $P$ is interesting, as it is part of both $C'$ and $C''$. On the other hand, $\tilde{P}$ is not interesting, as it is part of $C'$ but not of $C''$.  

\begin{lemma}~\cite{DBLP:conf/soda/DidimoLOP20,DBLP:journals/jgaa/RahmanNN03}.
$G$ has no bad 3-legged cycle if every demanding cycle contains a dummy along an edge of an interesting contour path.
\end{lemma} 

A demanding cycle $C\in D(G)$ is \emph{expensive} if and only if  the following two properties hold (refer to  \cref{app-fig:cubic-interesting-a} for an example).   

\begin{enumerate}[(C1)]
\item \label{enum:C1}
$C$ has exactly one interesting contour path $P$ and there exists at least a non-demanding 3-legged cycle $C'$ such that
$C'$ contains $C$, share edges with $P$ and with no other demanding cycle that it contains.
\item \label{enum:C2} 
Such contour path $P$ of $C$ is formed by an edge. 
\end{enumerate}

\begin{figure}[t]
  \begin{subfigure}{0.48\textwidth}
    \centering
    \includegraphics[page=5]{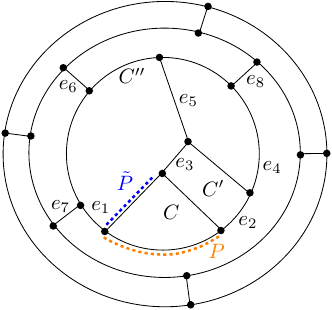}
    \subcaption{}
     \label{app-fig:cubic-interesting-a}
  \end{subfigure}
  \hfil
  \begin{subfigure}{0.48\textwidth}
    \centering
    \includegraphics[page=6]{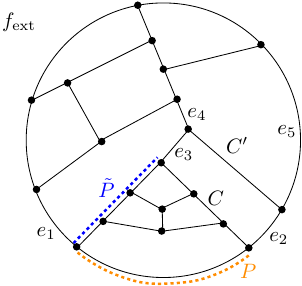}
    \subcaption{}
    \label{app-fig:cubic-interesting-b}
  \end{subfigure}
  \hfill
  \caption{(a) Illustration of the concept of interesting contour path and expensive demanding cycle; (b)~illustration of the concept of short demanding cycle.}
  \label{app-fig:cubic-interesting}
\end{figure}

See for example \cref{app-fig:cubic-cycles-a}, where $C_5$ and $C_9$ are expensive. We have that $C_4$ and $C_2$ are the the non-demanding cycle sharing exactly one edge with $C_5$ and $C_9$ and not sharing edges with other demanding cycles.

Let $D_{\exp}(G)\subseteq D(G)$ be the set of the expensive cycles of $G$. 
The following lemma refines the lower bound defined in \cref{app-co:cubic-lower1}.

\begin{lemma}
\label{app-le:cubic-lower2}
	$\tbn(G)\ge 2|D(G)|+|D_{\exp}(G)|+8-\min\{8, 2|D_{\ext}(G)|\}$.
\end{lemma}
\begin{proof}
Let $\mathcal{C}$ be an unbent collection with $b=\tbn(G)$ dummies.  We show that $\mathcal{C}$ has at least $2|D(G)|+|D_{\exp}(G)|+8-\min\{8, 2|D_{\ext}(G)|\}$ dummies. 
If $|D_{\exp}(G)|=0$ the lemma follows by 
\cref{app-co:cubic-lower1}. Suppose  $|D_{\exp}(G)|>0$.
Consider an expensive demanding cycle $C$. Let $e$ be the edge forming the only interesting contour path of $C$. Let $C'$ be the non-demanding cycle containing $C$, sharing edges with the interesting contour path of $C$ (i.e., $e$ is along the boundary of $C$ and $C'$), and not containing other non-demanding cycles with this property (see Property~\ref{enum:C1}). Cycle $C'$ does not share the dummy with $C$ in at least one of the representations, otherwise $e$ (which is the only edge they share) would not have a representation where it is not subdivided. It follows that, in such representation, in order to make both $C$ and $C'$ good cycles, which is necessary by \cref{app-le:cubic-badcycles}, we need two dummies, one for $C$ and one for $C'$. 
In all the other representations, we need one dummy, that is shared by $C$ and $C'$ and placed along $e$.
Notice that expensive cycles are contained in $D(G)$ and consequently they are vertex disjoint. Hence, the above consideration hold for every expensive cycle independently. It follows that, if $|\mathcal{C}|=2$ (since $G$ is triconnected, this is the minimum possible cardinality of $\mathcal{C}$), $\mathcal{C}$ has at least $|D_{\exp}(G)|$ dummies in addition to the ones of \cref{app-co:cubic-lower1}. Hence, the lemma follows.
\end{proof}

We now refine the lower bound of \cref{app-le:cubic-lower2}, introducing \cref{app-le:cubic-lower3}.  Before that we show the following lemma, that states a relationship between the contour paths of the cycles in $D(G)$ and the external face $f_{ext}$. In particular, a contour path $P$, that is a contour path of a cycle $C\in D(G)$ and sharing edges with $f_{\ext}$ ($f_{\ext}$ is a leg face of $C$), is always a sub-path of a contour path of a non-demanding cycle $C'$ containing $C$ (if $C'$ exists). Consequently, $P$ is always interesting. The same holds with respect to faces incident to $f_{\ext}$, when $f_{\ext}$ is not a leg face of $C$. More precisely, we have the following.

\begin{lemma}[\cite{DBLP:conf/soda/DidimoLOP20}]
\label{app-le:cubic_interestingexternal}
Let $C\in D(G)$ either: having $f_{\ext}$ as a leg face; or not having $f_{\ext}$ as a leg face and having a face $f'$ incident to $f_{\ext}$ as a leg face. One of the interesting contour paths of $C$ is incident to $f_{\ext}$ in the former case and to $f'$ in the latter case.
\end{lemma}

A cycle $C\in D_{\ext}(G)$ is \emph{short} if and only if it has (i)~exactly one edge $e$ incident to $f_{\ext}$ and (ii)~two or more interesting contour paths.  Observe that $e$ is a contour path (in particular, a contour path consisting of exactly one edge) of $C$ and that it is interesting by \cref{app-le:cubic_interestingexternal}. Let $D_{\short}(G)$ be the set of the short demanding cycles. Observe that $D_{\short}(G)\cap D_{\exp}(G)=\emptyset$, as (ii)~contradicts \ref{enum:C1}.  For an example, see \cref{app-fig:cubic-interesting-b}, where there are two 3-legged cycles $C$ (with legs $\{e_1,e_2,e_3\}$) and $C'$ (with legs $\{e_1,e_4,e_5\}$). 
Note that $C$ is demanding and that the contour paths~$P$ and~$\tilde{P}$ of $C$ are interesting, as they are part of $C'$. 
Also, $C$ is short, as it has two interesting contour paths ($P$ and $\tilde{P}$) and since $P$ is composed of one edge and incident to $f_{\ext}$. 

Thanks to the concept of short demanding cycle we can refine the lower bound of \cref{app-le:cubic-lower2} with the following lemma. Also, the following lower bound is tight, as we are going to prove in the next section.

\cubiclower*
\label{app-le:cubic-lower3*}

\begin{proof}
Let $\mathcal{C}$ be an unbent collection with $b=\tbn(G)$ dummies. If $2|D_{\ext}(G)|-|D_{\short}(G)|\ge 8$ or if $|D_{\short}(G)|=\emptyset$, then the bounds in \cref{app-le:cubic-lower2} and \cref{app-le:cubic-lower3} coincide. 
Hence, suppose $2|D_{\ext}(G)|-|D_{\short}(G)|\le7$ and $|D_{\short}(G)|\ge 1$. We have to show that in this case there are $|D_{\short}(G)|$ additive dummies in $\mathcal{C}$ with respect to the ones of \cref{app-le:cubic-lower2}. Consider $C\in D_{\short}(G)$ and let $e$ be the edge of $C$ incident to $f_{\ext}$, forming a contour path of $C$. Recall that, by \cref{app-le:cubic_interestingexternal}, such contour path is interesting. Hence, placing a dummy on $e$ is optimal, as in this case such dummy place the role of: A dummy for $C$ and all the non-demanding sharing an edge with $C$; one of the 4 dummies in the external face. On the other hand, in our setting it is not possible that there is a dummy on $e$ in all the representations of $\mathcal{C}$. Let $H\in \mathcal{C}$ be the representation where there is no dummy in the edge $e$ of $C$ incident to $f_{\ext}$. In such representation, we have to place a dummy in an edge of $C$ belonging to an interesting contour path of $C$ (which exists by definition of short cycles) not incident to $f_{\ext}$ and in another edge of the external face of $G$, instead of placing just on $e$. In all the others representation, just a dummy placed on $e$ suffices. Since $D_{\ext}(G)\le 7$ and since $C\in D_{\ext}(G)$, the cost with respect to the optimum increases of at least 1 in $H$. This consideration holds for each element of $D_{\short}(G)$. If $|\mathcal{C}|=2$ (which is the minimum we can have) $\mathcal{C}$ has at least $|D_{\short}(G)|$ dummies in addition to the ones of \cref{app-le:cubic-lower2} and, in particular, the lemma follows.
\end{proof}

\begin{figure}[t]
  \begin{subfigure}{0.35\textwidth}
    \centering
    \includegraphics[page=3]{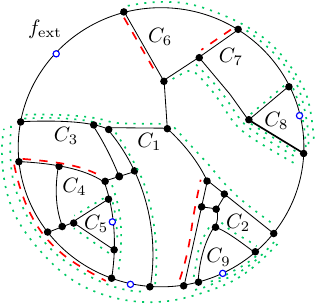}
    \subcaption{}
     \label{app-fig:cubic-cycles-1-a}
  \end{subfigure}
  \hfil
  \begin{subfigure}{0.35\textwidth}
    \centering
    \includegraphics[page=4]{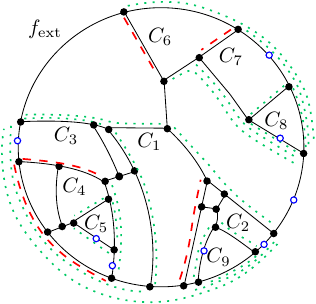}
    \subcaption{}
    \label{app-fig:cubic-cycles-1-b}
  \end{subfigure}
  \hfill
  \caption{An unbent collection of the graph $G$ depicted in
    \cref{app-fig:cubic-cycles-a} with dummies (blue circles) whose
    number matches the lower bound of \cref{app-le:cubic-lower3}.}
  \label{app-fig:cubic-cycles-1}
\end{figure}

 In \cref{app-subse:cubic-alg} we are going to show that the lower bound of \cref{app-le:cubic-lower3} is tight.  \cref{app-fig:cubic-cycles-1} shows an unbent collection of $G$ depicted in \cref{app-fig:cubic-cycles-a}, where dummies are represented as disks. As observed before, for such $G$ we have $D(G)=\{C_3,C_5,C_8.C_9\}$ and $D_{\ext}(G)=\{C_3,C_8.C_9\}$. We also have $D_{\exp}(G)=\{C_5,C_9\}$ and $D_{\short}(G)=\{C_8\}$. We have that such collection has overall $\tbn(G)= 2|D(G)|+|D_{\exp}(G)|+8-\min\{8, 2|D_{\ext}(G)|-|D_{\short}(G)|\}=8+2+8-5=13$ dummies.

\subsection{Computing a Collection with tbn(G) Bends}
\label{app-subse:cubic-alg}

In this section we show that the lower bound expressed in \cref{app-le:cubic-lower3} is tight and we prove \cref{app-th:minbend-cubic}. A 3-legged cycle is \emph{maximal} if it is not contained in another 3-legged cycle. We use the following lemma to prove \cref{app-le:cubic-setscomp}, which is a key ingredient for \cref{app-th:minbend-cubic}.

\begin{lemma}[\cite{DBLP:conf/soda/DidimoLOP20,DBLP:journals/jgaa/RahmanNN03}]
\label{app-le:cubic-demanding-comp} 
There exists an $O(n)$-time algorithm that computes: the maximal 3-legged cycles of $G$; all the demanding cycles of $G$; the set $D(G)$ of the non-pairwise intersecting ones and their interesting contour paths.
\end{lemma}

\begin{lemma}
\label{app-le:cubic-setscomp}
The sets $D_{\exp}(G)$ and $D_{\short}(G)$ can be computed in $O(n)$ time.
\end{lemma}
\begin{proof}
  By \cref{app-le:cubic-demanding-comp}, the set $D(G)$ can be
  computed in $O(n)$ time and also their interesting contour
  paths. Since they are not intersecting, they are independent. Hence,
  we can traverse all their edges. Whenever we find a contour path
  that is composed of an edge, we verify if it is an interesting
  contour path and if it is the only interesting contour path of the
  cycle. In this case, the cycle is expensive by definition and we add it to $D_{\exp}(G)$, otherwise
  it is not and we do nothing. Concerning $D_{\short}(G)$, we visit each cycle of
  $D(G)$ and we count the interesting contour paths of each one of
  them. For the ones with 2 or more interesting contour paths, we test
  if there exists one incident to~$f_{\ext}$ and if such contour path is composed of one edge.  If it exists, we add
  the cycle to $D_{\short}(G)$.
\end{proof}
 
We now turn to \cref{app-le:cubic-unbentT}, which is another key result to prove \cref{app-th:minbend-cubic}. Before that, we focus on the concept of a genealogical tree $T(C)$ for a 3-legged cycle~$C$.
In particular, we enrich $T(C)$ with more information, as follows.

\begin{lemma}[\cite{DBLP:journals/jgaa/RahmanNN03}]
\label{app-le:cubic-TCcomp}
Given a 3-legged cycle $C$ of $G$, $T(C)$ can be computed in $O(n_C)$ time. Also, no two 3-legged cycles of $G_C$ intersect.
\end{lemma}

We say that a non-demanding cycle $C'$ is \emph{critical} if, for every demanding cycle $C''$ it contains, $C'$ shares exactly one edge with $C''$ and such edge corresponds to an interesting contour path of $C''$. Cycles $C_5$, $C_6$, and $C_9$ in  \cref{app-fig:cubic-enrichedgen-a} are critical.

To efficiently place dummies, we define the \emph{enriched genealogical tree} $T^+(C)$, which is a data structure obtained from $T(C)$ by enriching it with additional information. More precisely, each (non-demanding) cycle $C'$ is equipped in $T^+(C)$ with a binary label~$\lambda(C')$ which is set to~$\mathit{true}$ if $C'$ is critical and to false otherwise. Then, each critical cycle $C'$ is also equipped with two pointers $\alpha(\cdot)$ and $\beta(\cdot)$ such that $\alpha(C')$ and $\beta(C')$ are two demanding cycles contained in $C'$. We say that such demanding cycles are \emph{associated with} $C'$. The cycles $\alpha(C')$ and $\beta(C')$ are associated with $C'$ so that the two following properties hold:

\begin{enumerate}[(P1)]
\item \label{enum:P1} 
 $\alpha(C')\not=\emptyset$ and, if $\beta(C')=\emptyset$, $\alpha(C')$ is expensive.
\item \label{enum:P2} Let $C''$ and $C'''$ be two critical cycles of $T(C)$ such that: $C'$ contains both $C''$ and $C'''$; $C''$ contains $C'''$. If $C'''$ and $C'$ are associated with $\tilde{C}$ and $C'''$ and $C''$ are associated with $\overline{C}$, then $C'$ is associated with $\overline{C}$.
\end{enumerate}

\begin{figure}[t]
  \begin{subfigure}{0.45\textwidth}
    \centering
    \includegraphics[page=10]{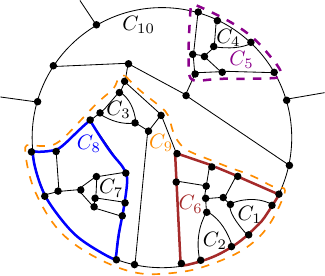}
    \subcaption{}
     \label{app-fig:cubic-enrichedgen-a}
  \end{subfigure}
  \hfil
  \begin{subfigure}{0.45\textwidth}
    \centering
    \includegraphics[page=11]{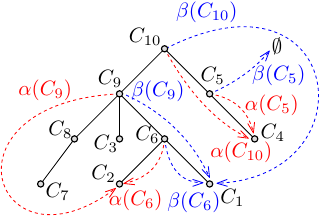}
    \subcaption{}
    \label{app-fig:cubic-enrichedgen-b}
  \end{subfigure}
  \hfill
  \caption{Illustration of enriched genealogical tree data structure.}
  \label{app-fig:cubic-enrichedgen}
\end{figure}

See for example \cref{app-fig:cubic-enrichedgen}. \cref{app-fig:cubic-enrichedgen-a} depicts $G_C$ and \cref{app-fig:cubic-enrichedgen-b} represents $T^+(G)$.  Observe that \ref{enum:P1} holds, as $\alpha(C)\not= \emptyset$ for any 3-legged cycle; also, we have that $\alpha(C_5)=C_4$ and $\beta(C_5)=\emptyset$ and $C_4$ is an expensive demanding cycle. Observe also that \ref{enum:P1} holds: For example, we have that $\alpha(C_9)=C_2$, $\beta(C_6)=C_1$ and, consequently, $\beta(C_9)=C_1$.

\begin{lemma}
  \label{app-le:cubic_TpC}
  There exists an $O(n_C)$-time algorithm that computes $T^+(C)$.
\end{lemma}

\begin{proof}  
We attach to the three degree-2 vertices in the external cycle of $G_C$ a 3-cycle so that $G_C$ is now cubic. Then, we compute $T(C)$ by \cref{app-le:cubic-TCcomp}, and $D(G_C)$ and the relative interesting contour paths by \cref{app-le:cubic-demanding-comp}. We then remove the external 3-cycle.

We first show a bottom-up procedure that adds for each contour path $P$ a label $\lambda(P)=\mathit{false}$ 
if an interesting contour path of a demanding cycle containing two or more edges is part of $P$, and $\lambda(P)=\mathit{true}$ otherwise. Also, if $\lambda(P)=\mathit{true}$ we associate $P$ with two pointers $\alpha(\cdot)$ and $\beta(\cdot)$ such that $\alpha(P)$ and $\beta(P)$ are two demanding cycles whose controur paths are contained in $P$. If $P$ does not contain contour paths of demanding cycles, then $\alpha(P)=\beta(P)=\emptyset$, while if $P$ contains only one such contour path of demanding cycle, then $\alpha(P)$ is the demanding cycle, and $\beta(P)=\emptyset$. 

Let $C'$ be the first non-demanding cycle that we consider. Observe that $C'$ has only demanding cycles as children and that no edge of it was previously visited by the procedure. Let $P$ be one of its contour paths. We traverse $P$ from its endpoint $u$ to its other endpoint $v$. If along $P$ there is an interesting contour path of a demanding cycle that has 2 or more edges, $\lambda(P)=\mathit{false}$. Otherwise, we associate to $\alpha(P)$ and $\beta(P)$ the (up to) two demanding cycles as described before. We do that for all the contour paths of $C'$. Suppose we are now at step $i$ ($i>1$) of the procedure and that we have the information described above for every non-demanding cycle analyzed in the previous steps. Let $C'$ be the currently visited non-demanding cycle. 

We first add some information on the legs of the children of $C'$, so that we can traverse the edges of $C'$ only if they were not traversed before. Namely, for each child $C^*$ of $C'$ and each contour path $P^*$ of $C^*$, we test if $P^*$ is part of any contour path of $C'$ (which we can check in $O(1)$ time by checking if $C'$ and $C^*$ have the same leg face in the two contour paths). In this case, we consider the legs of $P^*$ and for each one of them we add a pointer to $P^*$ and to the other leg (i.e., the two legs are pointing each other). We can now consider a contour path $P$ of $C'$ and visit its edges from $u$ to $v$, where $u$ and $v$ are the two endpoints of $P$. These pointers let us traverse $P$ so that when we encounter a leg $l$ of a child of $C'$, we analyze the contour path $P^*$ pointed by $l$ (just checking the values $\alpha(P^*)$ and $\beta(P^*)$, as we will describe later) and then we proceed our traversal of $P$ from the other leg pointed by $l$ (not visiting the edges of $P^*$). 

Given this traversal strategy, if while traversing we encounter  an interesting contour path of a demanding cycle with 2 or more edges (notice that in this case such cycle is a child of $C'$) or a contour path $P^*$ of a child of $C'$ so that $\lambda(P^*)=\mathit{false}$, we set $\lambda(P)=\mathit{false}$.
Otherwise, we associate with $P$ the cycles $\alpha(P)$ and $\beta(P)$: If $C'$ has a demanding cycle as a child, we associated that with $\alpha(P)$; else, we consider a contour path $P^*$ of a child $C^*$ of $C'$ and we set assign to $\alpha(P)$ the value $\alpha(P^*)$, if it is not $\emptyset$, or $\beta(P^*)$, if it is not $\emptyset$; finally (if such $C^*$ and $P^*$ do not exist), we set $\alpha(P)=\emptyset$. We do the same for $\beta(P)$, making sure that $\alpha(P)\not= \beta(P)$ or $\alpha(P)= \beta(P)=\emptyset$.

At the end of this first bottom-up traversal, we consider each non-demanding cycle $C'$ of $T(C)$. If $\lambda(P)=\mathit{true}$ for all the three contour paths of $C'$, then $C'$ is critical and we set $\lambda(C')=\mathit{true}$. We now associate demanding cycles with $C'$ so that \ref{enum:P1} and \ref{enum:P2} hold. We perform another bottom-up visit of $T(C)$. During the bottom-up procedure, suppose we are choosing how to set $\alpha(C')$. We proceed as follow:

\begin{enumerate}
\item We first check if there is a demanding cycle $C$ that is a child of $C'$.
In this case, we set $\alpha(C')=C^*$. 
\item Otherwise, we check if there is a child $C^*$ of $C'$ sharing a contour path with $C'$  and associated with a demanding cycle $\tilde{C}$ contained in such contour path ($\tilde{C}=\alpha(C^*)$ or $\tilde{C}=\beta(C^*)$): in this case, we set $\alpha(C')=\tilde{C}$.
\item If the cases above do not apply, we check if there exists a contour path $P^*$ of a child $C^*$ of $C'$ sharing edges with $C'$ so that there exists a demanding cycle $\tilde{C}$ such that $\alpha(P^*)= \tilde{C}$ or $\beta(P^*)=\tilde{C}$  and we set $\alpha(C')=\tilde{C}$ (in this case, $\tilde{C}$ was not associated with $C^*$, but only with $P^*$). 
\end{enumerate}

We set $\beta(C')$ similarly, making sure that $\beta(C')\not = \alpha(C')$ or $\beta(C')= \alpha(C')=\emptyset$. Concerning \ref{enum:P1}, as $C'$ is non-demanding we have that it shares edges with at least a demanding cycle, that is either a child of it (Association~1), or contained in a contour path of a child of $C'$ (Association~2-3). Hence, $\alpha(C')\not = \emptyset$. As we always selected cycles having an interesting contour path that is part of $C'$, we have that if $C'$ can only be associated to a cycle $\alpha(C')$, then $\alpha(C')$ is expensive by definition. It remains to show that  \ref{enum:P2} holds. If we performed Association~1 or Association~3, then the property trivially holds, as the cycles $C''$ and $C'''$ mentioned by the proof do not exist.
Suppose we performed Association~2. In this case, we associated to $C'$ a demanding cycle that was previously associated to a child $C''$ of it and, consequently, in this case does not exists $C'''$ that is a successor of $C''$ and such that $C'$ are $C'''$ associated to a same cycle and $C''$ is not associated to it.

By \cref{app-le:cubic-TCcomp} two cycles are never intersecting in $T(C)$. Hence, every edge of $G_C$ is visited $O(1)$ times thanks to the pointers that we add to the legs of the non-demanding children of every currently visited cycle. It follows that the enriched genealogical tree $T^+(C)$ can be computed in $O(n_C)$ time.
\end{proof}

\begin{lemma}
\label{app-le:cubic-useAleg}
Let $C$ be a demanding 3-legged cycle, let $P$ be an interesting contour path of $C$. There is a leg incident to a vertex of $P$ that lies along $C'$ and $C''$, where $C'$ and $C''$ are any two non-intersecting non-demanding cycles containing $C$ and sharing edges with an interesting contour path of $C$.
\end{lemma}
\begin{proof}
See \cref{app-fig:cubic-interesting-a} for an example of the statement of the lemma, where $P$ is an interesting path, $e_2$ is leg of the demanding cycle $C$, and $e_2$ is also and edge of $C'$ and of $C''$. Let $l_1$ and $l_2$ be the two legs of $C$ incident to vertices of $P$. Since $C'$ and $C''$ are non-intersecting, it is not possible that we both have that: $l_1$ is an edge of $C'$ and not an edge of $C''$; and that $l_2$ is an edge of $C''$ and not an edge of $C'$. 

Given the observation above, it suffices to prove that for every cycle $C'$ containing $C$ and sharing edges with $P$, there exists  $l\in \{l_1,l_2\}$ that is an edge of $C'$. Suppose by contradiction that it is not the case. Let $v_1$ and $v_2$ be the vertices of $C$ incident to $l_1$ and $l_2$, respectively. Since $P$ is part of $C'$, we have that $v_1$ and $v_2$ are part of $C'$, but not $l_1$ and $l_2$. Hence, both $l_1$ and $l_2$ are legs of $C'$. Let $P_1$ and $P_2$ be the contour paths of $C$ distinct from $P$ and incident to $v_1$ and $v_2$, respectively. We have that $P_1$ and $P_2$ cannot be both part of $C'$, since $P$ is part of $C'$ (and $C$ and $C'$ are distinct). Suppose without loss of generality $P_1$ is not part of $C'$.  There is another contour path $P'$ of $C'$ incident to $v_1$ and to $l_1$, as $l_1$ is a leg of $C'$ as observed before. Hence, $v_1$ is incident to a vertex of $P$, $P'$, $P_1$, and $l_1$, which are all distinct. A contradiction.

Concerning the computational time, such leg is one of the two legs of $C$ incident to $P$ and that is not a leg of $C'$, which can be computed in $O(1)$ time as for any 3-legged cycle we can assume to know its three legs.
\end{proof}

Let $C$ be a be a 3-legged cycle of $G$ such that $f_{\ext}$ is a leg face of $C$. The contour path of $C$ incident to $f_{\ext}$ is the \emph{exposed} contour path of $C$.

\begin{lemma}
\label{app-le:cubic-unbentT}
Let $C$ be a 3-legged cycle of $G$ and $n_C$ be the number of vertices of $G_C$. There exists an $O(n_C)$-time algorithm that computes a collection $\mathcal{C}$ of $G_C$ such that the following properties hold:
\begin{enumerate}[(i)]
    \item the collection has $2|D(G_C)|+|D_{\exp}(G_C)|$ dummies and no bad cycle.
    \item If $C$ has an exposed contour path $P$, this path  is subdivided with $2t-|D_{\short}(G_C)|$ dummies incident to $f_{\ext}$, where $t$ is the number interesting contour paths of 3-legged cycles along $P$.
\end{enumerate} 
\end{lemma}

\begin{proof}
We compute the enriched genealogical tree $T^+(C)$ in $O(n_C)$ time by \cref{app-le:cubic_TpC}. We also compute $D(G_C)$, $D_{\exp}(G_C)$, and their interesting contour paths by means of~\cref{app-le:cubic-demanding-comp} and \cref{app-le:cubic-setscomp}. We consider two copies of $G_C$, namely $G^1_C$ and as $G^2_C$ and place dummies along their edges (i.e., we compute $\mathcal{C}$ such that $|\mathcal{C}|=2$).
The proof consists of two parts: In Part 1 we describe an  $O(n_C)$-time procedure that subdivides edges of $G^1_C$ and $G^2_C$ with $2|D(G_C)|+|D_{\exp}(G_C)|$ dummies so that neither $G^1_C$ nor $G^2_C$ has a bad $3$-legged cycle. In this part we prove both (i) and (ii). In Part 2 we prove that at the end of this subdivision procedure neither $G^1_C$ nor $G^2_C$ contains a bad cycle (observe that in this case a bad cycle could be the external cycle or a 2-legged cycle generated by the placement of the dummy vertices, see for example $C_1$ and $C_2$ in \cref{app-fig:2leggedbydummies}). Finally, we discuss the computational time.

\smallskip \noindent {\sffamily\bfseries Part 1.} We first focus on $(i)$.
The strategy is counting the number of dummies that we place in both $G^1_C$ and $G^2_C$ and that exceed the number $2|D(G_C)|$ (for example, if analyzing a cycle in $D(G_C)$ we place 3 dummies overall in the collection, we are exceeding 1).  Observe first that $G_C$ has no intersecting cycles by \cref{app-le:cubic-TCcomp}. Hence, we can focus on pairwise non-intersecting cycles. Consider each demanding cycle in $T^+(C)$ with at least two edges belonging to an interesting contour path. We consider two edges of it following this rule:

\smallskip
\noindent
\underline{Non-expensive edge-selection rule}: If there exists two of such edges are in the external face, we consider these two. If there exists just one (the cycle is short), consider this one. Otherwise, we consider any two edges of the cycle belonging to an interesting contour path. 

\smallskip
Given the two edges selected with the non-expensive edge-selection rule, we subdivide one of them in $G^1_C$ and the other one in $G^2_C$. We now have that every non-expensive demanding cycle of $G$ is not a bad cycle in $G^1_C$ and it is not a bad cycle in $G^2_C$. Consider now the demanding cycles  having only one interesting contour path which consists of exactly one edge and the critical cycles sharing edges with them (some of these demanding cycles are expensive). In order to count the number of dummies exceeding $2|D(G_C)|$, we perform a top-down traversal of $T^+(C)$ where we consider only the critical cycles (i.e., if there are $k$ critical cycles, our counting is executed in $k$ steps). Recall that we can test if a cycle $C'$ is critical if $\lambda(C')=\mathit{true}$.  We set $count_0=0$. At step $i$ ($i\in [i,k]$), let $count_i$ be a counter that we increment by $1$ with respect to $count_{i-1}$ whenever we place a dummy in the critical cycle that is not part of the interesting contour path of a demanding cycle. If only interesting contour paths are subdivided we obtain $count_k=0$. 
We show how to place dummies such that $count_k=|D_{\exp}(G_C)|$ and that there is no bad 3-legged cycle. 
 
When we consider a critical cycle $C'$, let $e_x$ be the edge forming the interesting contour path of $\alpha(C')$ contained in $C'$. 
If $\beta(C')\not = \emptyset$, then we define $e_y$ similarly. During the procedure, we implicitly assume that if we subdivide $e_x$ in $G^a_C$ ($a\in \{1,2\}$), we subdivide another edge of a contour path of $\alpha(C')$ that is not interesting in $G^b_C$ ($b\in \{1,2\}\setminus a$), so that after this placement $\alpha(C')$ is not a bad cycle in both $G^1_C$ and $G^2_C$ (in some case we explicitly say how to perform also this subdivision, if this is helpful for the description of the algorithm). We do that also for $\beta(C')$ when it is not equal to $\emptyset$. During the procedure, we maintain three invariants:
\begin{enumerate}[({I}1)] 
\item \label{enum:invar-A} 
If $\beta(C')\not = \emptyset$, we do not subdivide both $e_x$ and $e_y$ in the same graph between $G^1_C$ and $G^2_C$;
\item \label{enum:invar-B} 
If $e_z$ ($z\in \{x,y\}$) was subdivided in a previous step in $G^a_C$ ($a\in \{1,2\}$), we do not subdivide $e_x$ or $e_y$ in $G^a_C$;
\item \label{enum:invar-C}
Whenever we consider $e_x$ or $e_y$ for the first time, we perform subdivision operations so that none of the critical cycles containing $e_x$ or $e_y$ is bad in both $G^1_C$ and $G^2_C$.
\end{enumerate}
Consider the case $i=1$. If $\beta(C')\not = \emptyset$, we subdivide $e_x$ in $G^1_C$ and $e_y$ in $G^2_C$ and we set $count_1=count_0=0$. Hence, \ref{enum:invar-A} holds. Otherwise, we subdivide $e_x$ in $G^1_C$, while for $G^2_C$ we consider the leg $l$ of $\alpha(C')$ whose existence is stated in \cref{app-le:cubic-useAleg} and any edge $e_{\tilde x}$ of $\alpha(C')$ different from $e_x$. We subdivide both of them and, by \cref{app-le:cubic-useAleg}, none of the non-demanding cycles sharing $e_x$ with $\alpha(C')$ are bad in both $G^1_C$ and $G^2_C$. By the same lemma invariant  \ref{enum:invar-C} holds. In this case $count_1=count_0+1=1$, as we placed two dummies instead of one considering $\alpha(C')$ in $G^2_C$ (one dummy in $e_{\tilde x}$ and one in $l$). Observe that in this case \ref{enum:invar-B} trivially holds as this is the first step.

Consider now the case $i>1$ and suppose that invariants~\ref{enum:invar-A}, \ref{enum:invar-B}, and \ref{enum:invar-C} were always preserved during the previous steps. If $\beta(C')= \emptyset$, we proceed as in the case $i=1$ with the difference that if $\alpha(C')$ was previously subdivided, we do nothing, so that \ref{enum:invar-B} holds. The other two invariants hold with the same argument of the case $i=1$. Suppose $\beta(C')\not = \emptyset$. We first prove that $e_x$ and $e_y$ are not both subdivided in one between $G^1_C$ and $G^2_C$. We prove it does not happen in $G^1_C$. The same proof can be applied for $G^2_C$. 

Suppose by contradiction that $e_x$ and $e_y$ are both subdivided in one between $G^1_C$ and $G^2_C$. By \ref{enum:invar-A}, we have that we did not subdivide them in the same graph in step $j\in[1,j-1]$. Hence, they were subdivided in two different steps, while analyzing two different cycles $C_p$ and $C_{p'}$. As they share at least one edge and by \cref{app-le:cubic-TCcomp}, either $C_p$ contains $C_{p'}$ or vice versa. Suppose $C_{p'}$ contains $C_p$. By \ref{enum:P2}, if $C_{p'}$ was associated to $e_x$ or $e_y$, the same holds for $C_p$. By \ref{enum:invar-B}, if one between $e_x$ and $e_y$ was previously subdivided in $G^1_C$ by $C_{p'}$, $C_p$ did not subdivide the other one in $G^1_C$. A contradiction. Hence, $e_x$ and $e_y$ are not both subdivided in one between $G^1_C$ and $G^2_C$.

Since $e_x$ and $e_y$ are not subdivided in the same graph between $G^1_C$ and $G^2_C$, we can subdivide them so that one is subdivided in $G^1_C$ and one in $G^2_C$. (For example, if $e_x$ was subdivided in $G^2_C$, we subdivide $e_y$ in $G^1_C$.) Observe that with this strategy, \ref{enum:invar-A} and \ref{enum:invar-B} hold. Concerning \ref{enum:invar-C}, it holds as we subdivided an interesting contour path sharing edges with $C'$ in both $G^1_C$ and $G^2_C$. We set $count_{i}=count_{i-1}$, as we subdivided each demanding once in both $G^1_C$ and $G^2_C$. This concludes the description of the placement of the dummies.

After this procedure, $G^1_C$ and $G^2_C$ have no bad 3-legged cycle as, for every demanding cycle, we placed a dummy either in one of its interesting contour paths or in the edge individuated by \cref{app-le:cubic-useAleg} in both $G^1_C$ and $G^2_C$. We placed $2|D(G_C)|+count_k$ dummies.
We increment $count_i$ with respect to $count_{i-1}$ ($i\in [1,k]$) only when $\beta(C')=\emptyset$. By \ref{enum:P1} of the enriched genealogical tree $T^+(C)$, we placed exactly $2|D(G_C)|+|D_{\exp}(G_C)|$ subdivision vertices.

We now show that the placement satisfies (ii). We consider the cycles that are contained in $D_{\short}(G_C)$. Observe that, each $C^*\in D_{\short}(G_C)$ has at last two edges belonging to interesting contour paths. Also, only one edge $e^*$ lies along $P$. Hence, according to the non-expensive edge-selection rule, we subdivided such edge $e^*$ in exactly one between $G^1_C$ or $G^2_C$.
We now have that every cycle in $D_{\short}(G_C)$ has exactly one dummy contained in $P$ in the collection. Hence, at most $2t-|D_{\short}(G_C)|$ dummies are incident to $f_{\ext}$. We show now that exactly such number of dummies are incident to $f_{\ext}$. Namely, it remains to show that the other exposed interesting contour paths of non-short demanding cycles always imply a subdivision vertex on $f_{\ext}$. Consider critical cycles: in this case, in each one of $G^1_C$ and $G^2_C$, we either subdivided an exposed contour path contained in $P$ or (if $C'$ is expensive) we subdivided one of its legs incident to $f_{ext}$ following the property of \cref{app-le:cubic-useAleg}. For the demanding cycles with two edges contained in an interesting contour path that are not short and with an exposed contour path, we subdivided an edge incident to $f_{\ext}$ in both $G^1_C$ and $G^2_C$ by the non-expensive edge-selection rule. Hence, the collection has exactly $2t-|D_{\short}(G_C)|$ edges placed on the exposed contour path $P$ of $C$.

\smallskip \noindent
{\sffamily\bfseries Part 2.} Here we show that the dummies placed above are sufficient for $G^1_C$ and $G^2_C$ to not to have bad cycles, the remaining of which are the 0-legged cycle that is the external face of $G_C$ and the 2-legged cycles caused by the presence of dummies. We prove it for $G^1_C$, the proof for $G^2_C$ is the same.  Observe that $G_C$ is a plane triconnected cubic graph with three different edges of the external face that are subdivided. These vertices are the ones incident to the legs of $C$ in $G$ and here we can consider them as dummies; let $d_1$, $d_2$, and $d_3$ denote them. 

We first observe that $d_1$, $d_2$, and $d_3$ are not all in a same edge of $G_C$. Suppose they are, then $G$ contains the path $\{d_1,d_2,d_3\}$ (after a potential re-labeling); in this case, $d_1$ and $d_3$ are a separating pair of $G$.  

In the procedure of Part 1 of the proof we subdivided edges that were contained in interesting contour paths, i.e., contour paths of bad cycles of 3-legged cycles. These contour paths had no edge subdivided and consequently we did not subdivide any edge having one between $d_1$, $d_2$, or $d_3$. Hence, there are at least 3 edges subdivided in $G^1_C$ and by \cref{app-le:2legged}, there are no bad 2-legged cycles. 

We now prove that there are no pairwise intersecting demanding cycle contained in $G_C$. Suppose by contradiction that it is the case. Then, we slightly modify graph $G$ so that we have a contradiction to Condition (3) of \cref{app-le:cubic-only_intersecting}. We proceed as follow: we consider any edge not contained in $G_C$ and incident to $f_{\ext}$ and we subdivide it twice inserting two vertices $v_1$ and $v_2$; we subdivide any other edge incident to the same internal face with a vertex $v_3$; We add a vertex $v_4$ and the edges $v_1v_4$, $v_2v_4$, and $v_3v_4$ so that the graph we obtain is plane. The obtained graph is also plane triconnected cubic and it has a 3-legged cycle formed by three edges $\{v_1v_2, v_1v_4,v_2v_4\}$ incident to the external face of $G$ and not contained in $G_C$. Such 3-legged cycle is a demanding cycle, since no vertex is contained in its interior. Since $I(G)\not=\emptyset$ by hypothesis, we have a contradiction to Condition (3) of \cref{app-le:cubic-only_intersecting}. 

We have that after the subdivision of the 3-legged cycles of $G_C$ the external face had at least 4 dummies, which are $d_1$, $d_2$, $d_3$, and the one that had to be placed so that $C$ is not a bad cycle.  Hence, $G^1_C$ has no bad cycle. 

\smallskip
Concerning the computational time, for each critical cycle of $T^+(C)$ we only analyzed its children in $T^+(C)$ and the demanding cycles pointed by them for the critical ones. Hence, overall the procedure takes $O(n_C)$ time.
\end{proof}

We are now ready to prove the main theorem of this section.

\minbendcubic*
\label{app-th:minbend-cubic*}

\begin{proof}
We prove that we can compute in $O(n)$ time a collection with at most $q=2|D(G)|+|D_{\exp}(G)|+8-\min\{8, 2|D_{\ext}(G)|-|D_{\short}(G)|\} $ bends. Our strategy is first creating two copies of $G$, $G^1$ and $G^2$, and then placing on them overall exactly $q$ dummies. We prove that none of the two graphs has a bad cycle; then, by \cref{app-le:cubic-badcycles} and since $q\le \tbn(G)$ by \cref{app-le:cubic-lower3}, this directly gives us an unbent collection with $\tbn(G)$ bends. Let $f_{\ext}$ be the external face of $G$. We first prove the theorem for a simple case, where there is no 3-legged cycle with edges incident to $f_{\ext}$ (hence, in particular, $D_{\ext}(G)=\emptyset$) and $f_{\ext}$ is incident to at least four edges. If these conditions hold, we say that $G$ is \emph{nice}. 

\smallskip \noindent
{\sffamily\bfseries G is nice.} Let $G^1$ and $G^2$ be two copies of $G$. Observe that in this case, since $G$ is nice, $D_{\ext}(G)=\emptyset$ and consequently $2|D_{\ext}(G)|-|D_{\short}|=0$. Hence, by \cref{app-le:cubic-lower3}, $\tbn\ge q=2|D(G)|+|D_{\exp}(G)|+8$. We show how to place $q$ dummies overall so that $G^1$ and $G^2$ have no bad cycle.

We compute the maximal 3-legged cycles of $G$ by \cref{app-le:cubic-demanding-comp}.  Let $C$ be one of such maximal 3-legged cycles.  
Consider the subgraphs $G^1_C$ and $G^2_C$ of $G^1$ and $G^2$, respectively. 
We place dummies along the edges of $G^1_C$ and $G^2_C$ by the algorithm whose existence is stated by \cref{app-le:cubic-unbentT}. We perform this procedure for every maximal 3-legged cycle. We have that since $G$ is nice and by \cref{app-le:cubic-only_intersecting}(3), the maximal 3-legged cycles of $G$ are disjoint. Hence, by \cref{app-le:cubic-unbentT}, we placed exactly $q'=2|D(G)|+|D_{\exp}(G)|$ dummies along the edges of $G^1$ and of $G^2$.  Consider the external face of $G$ and select four edges $e_1^1,e_2^1,e_1^2,e_2^2$ incident to $f_{\ext}$ such that both the pair $e_1^1, e_2^1$ and the pair $e_1^2,e_2^2$ are independent pairs of edges (this is always possible, as $G$ is nice and consequently it has at least four edges incident to the external face). 
We place two dummies along $e_1^1$ and $e_2^1$ in $G^1$ and two dummies along $e_2^1$ and $e_2^2$ in $G^2$. The external cycle of both $G^1$ and $G^2$ is not a bad cycle, as we just placed four dummies in each of them. Also, there is no other bad cycle by \cref{app-le:2legged}.
 Observe that this procedure places $q=2|D(G)|+|D_{\exp}(G)|+8$ dummies.  The linear time complexity follows from \cref{app-le:cubic-badcycles}, by the fact that when $G$ is nice maximal 3-legged cycles of~$G$ are independent (as observed before), and by \cref{app-le:cubic-unbentT}.

\smallskip \noindent
{\sffamily\bfseries G is not nice.} We now consider the other cases where the situation is more difficult and where a more detailed discussion is required. We use as a subroutine of the next cases the placement of the case when $G$ is nice, in order to place dummy vertices in internal edges (i.e., edges not incident to $f_{\ext}$). We call such a placement strategy a \emph{nice-placement}. 
The next cases are the following.
\begin{enumerate}[(i)]
\item {\sffamily\bfseries Intersecting Case:} There are intersecting demanding cycles. 
\item {\sffamily\bfseries 3-Cycle Case:} The external face has 3 edges. 
\item {\sffamily\bfseries External-Demandings Case:} $D_{\ext}(G)\not=\emptyset$. 
\end{enumerate}

\begin{figure}[t]
  \begin{subfigure}{0.45\textwidth}
    \centering
    \includegraphics[page=15]{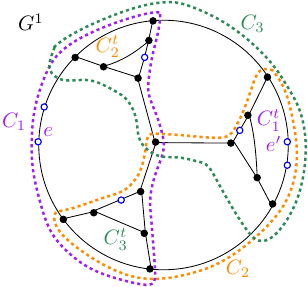}
    \subcaption{}
     \label{app-fig:intersecting-a}
  \end{subfigure}
  \hfil
  \begin{subfigure}{0.45\textwidth}
    \centering
    \includegraphics[page=16]{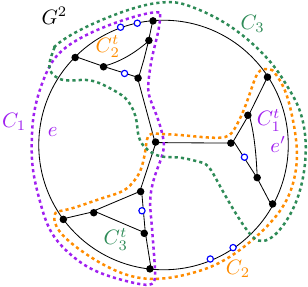}
    \subcaption{}
     \label{app-fig:intersecting-b}
  \end{subfigure}
  \hfil
  \caption{Illustration for the proof of \cref{app-th:minbend-cubic}
    when there are intersecting cycles.}
  \label{app-fig:intersecting}
\end{figure}

We first observe that these cases are mutually exclusive. The fact that (i)
 and (iii) are mutually exclusive directly comes from \cref{app-le:cubic-only_intersecting}(3). Concerning (ii) and the other two, it suffices to prove that when the external cycle is a 3-cycle there is no 3-legged incident to $f_{\ext}$. Suppose by contradiction that $f_{\ext}$ is incident to 3 edges and that there is a 3-legged $C$ having $f_{\ext}$ as leg face. Then two of the edges incident to $f_{\ext}$ are legs of $C$ and at least one of them is an edge of $C$ forming one of its contour paths. Since these are 3 edges already, incident to $f_{\ext}$, it means that all the 3 legs of $C$ are incident to a same vertex. A contradiction (see the definition of $k$-legged cycle).
 
\smallskip
\noindent
{\sffamily\bfseries (i) Intersecting Case.} Refer to \cref{app-fig:intersecting}, where $C_1$, $C_2$ and $C_3$ are intersecting demanding cycles and $C_1^t$, $C_2^t$, and $C_3^t$ are their respective twins.  We compute all the intersecting demanding cycles by \cref{app-le:cubic-demanding-comp}, computing all the demanding and subtracting from this set the demandings in $D(G)$. We select two intersecting demanding cycles $C_1$ and $C_2$ and the relative twins $C_1^t$ and $C_2^t$. We then consider $C_1$ and $C_1^t$, we select two edges $e\in C_1\cap f_{\ext}$ and $e'\in C_1^t\cap f_{\ext}$. We place in $G^1$ two dummies in $e$ and two dummies in $e'$. By \cref{app-le:cubic-only_intersecting}(2) there is no 3-legged bad cycle that is intersecting in $G^1$. Also the external face is now not a bad cycle in $G^1$ and by \cref{app-le:2legged} there is no bad 2-legged cycle. We do the same considering $C_2$, $C_2^t$, and $G^2$. Concerning the internal 3-legged cycles, we handle them with the nice-placement.
By \cref{app-le:cubic-only_intersecting}(3) we have $D_{\ext}(G)=\emptyset$.  Since $D_{\ext}(G)=\emptyset$, as in the case when $G$ is nice, we have $\tbn(G) \ge q=2|D(G)|+|D_{\exp}(G)|+8$ by \cref{app-le:cubic-lower3}.  We have that with this placement strategy we placed $q'=2|D(G)|+|D_{\exp}(G)|$ internal dummies by \cref{app-le:cubic-unbentT} as in the case when $G$ is nice and, considering the $8$ we placed in the external face to handle the intersecting demanding cycles, we placed overall $q=q'+8=2|D(G)|+|D_{\exp}(G)|+8$. Since we have no bad cycle, as we handled the intersecting applying  \cref{app-le:cubic-only_intersecting}(2) and since the 3-legged are handled by the same procedure used when $G$ is nice, the lemma follows. Concerning the computational time, the fact that this procedure takes linear time comes from \cref{app-le:cubic-demanding-comp} and since the nice-placement takes linear time (as proved in the case when $G$ is a nice graph).

\begin{figure}[t]
  \begin{subfigure}{0.45\textwidth}
    \centering
    \includegraphics[page=17]{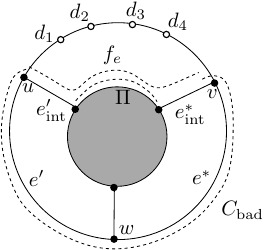}
    \subcaption{$G^1$}
     \label{fig:cubic-collection-3cycle-a}
  \end{subfigure}
  \hfill
  \begin{subfigure}{0.45\textwidth}
    \centering
    \includegraphics[page=18]{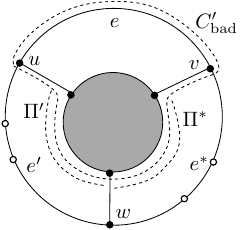}
    \subcaption{$G^2$}
    \label{fig:cubic-collection-3cycle-b}
 \end{subfigure}
  \hfill
  \caption{Illustration for the proof of \cref{app-th:minbend-cubic} when the external face is a 3-cycle in the 1-2-placement case.}
  \label{fig:cubic-collection-3cycle}
\end{figure}

\smallskip
\noindent
{\sffamily\bfseries (ii) 3-Cycle Case.} As observed while showing that the three cases we are considering are mutually exclusive, we have that in this case there is no 3-legged cycle having $f_{\ext}$ as leg face and consequently no intersecting cycle and $D_{\ext}(G)=\emptyset$.   We now have 3 edges in the external face and 4 dummies to place so that it is not a bad cycle, but here, the situation changes with respect to the case when $G$ is nice, since we cannot directly use \cref{app-le:2legged}. In fact, since the external face has 3 edges we cannot place the 4 dummies in 3 different edges or in at least 2 edges with 2 dummies each that are independent, as described by \cref{app-le:2legged}. Hence, here we do not use directly \cref{app-le:cubic-lower3}, but we rely on the fact that the nice-placement places dummies optimally in $G^1\setminus e$ and $G^2\setminus e$ and on the fact that the placement of dummies on edges of $f_{\ext}$ has bounded number of possibilities. Namely, there are three possibilities: 
\begin{enumerate}[(a)]
\item The 1-2-placement: The four dummies in the external face of $G^1$ are in a same edge $e$, while in $G^2$ they are in two edges, $e'$ and $e^*$; 
\item The 2-1-placement: Symmetric to the 1-2-placement;
\item The 1-1-placement: The four dummies in the external face of $G_1$ and $G_2$  are in a same edge $e$ and $e'$, respectively. 
\end{enumerate}
We first show how to handle the 1-2-placement. Consider first $G^1$. Refer to \cref{fig:cubic-collection-3cycle-a}. We subdivide $e=uv$ with four dummies $d_1,...,d_4$ so that $d_1$ is adjacent to $u$ and $d_4$ is adjacent to $v$. Let $C_{\bad}$ be the external cycle of $G^1\setminus e$. We have that $C_{\bad}$ is a bad 2-legged cycle and we need to place two dummies on it. We are going to show how to subdivide $C_{\bad}$ optimally. Observe that $C_{\bad}$ has two legs, which are $(u,d_1)$ and $(v,d_4)$. Suppose that $e'$ and $e^*$, which are the two edges different from $e$ incident to $f_{\ext}$, are incident to $u$ and $v$ respectively. Let $e'_{\mathrm{int}}$ and $e^*_{\mathrm{int}}$ be the other edges incident to $u$ and $v$, which are internal edges of $G$ . Let $\Pi$ be a path such that $C_{\bad}=e'\cup e^*\cup e^*_{\mathrm{int}} \cup e'_{\mathrm{int}} \cup \Pi$ and let $f_e$ be the face incident to $e$ and internal in $G$. 

We now use the nice-placement in $G\setminus e$. We just placed $q'=2|D(G)|+|D_{\exp}(G)|$.
By \cref{app-le:cubic_interestingexternal} the interesting contour paths and, consequently, the legs of the demanding cycles incident to $f_e$ lie along $\Pi$. Hence, if such demanding cycles exists, nice-placement subdivided $\Pi$ already by \cref{app-le:cubic-unbentT}(2). If $C_{\bad}$ has 2 dummies in $\Pi$, we do nothing, and the cost is optimum as we did not add dummies with respect to the nice-placement. Otherwise, suppose that that $C_{\bad}$ does not have 2 dummies. 

If $\Pi$ has 1 dummy, we place $d$ in the same edge and now $C_{\bad}$ is not a bad cycle anymore, as it has two legs and the two dummies in $\Pi$ (one of which is the newly-placed~$d$). 
Otherwise, $\Pi$ has no dummy and consequently it is not incident to any demanding interesting contour path. 
It follows that there is no dummy in $\Pi$ both in $G^1$ and $G^2$ after the internal placement, and we can place $d$ in any edge of $\Pi$. 
Observe that in this case we place up to 2 dummies more with respect to the nice-placement. 
Since in this case we are forced to place 4 dummies in $e$ and since $\Pi$ does not have an interesting contour path and by~$C_{\bad}$, we placed the minimum number of subdivision vertices.

We now consider $G^2$. Refer to \cref{fig:cubic-collection-3cycle-b}. Let $w$ be the vertex of $G$ incident to $e'$ and $e^*$. We have to consider the degenerate 3-legged cycle $C_{\bad}'$ of $G$ whose three legs are incident to $w$, since two of its legs are incident to dummies. In this case we have to place just one dummy $d'$. As before, we can observe that either the nice-placement placed such dummy in one of its contour paths $\Pi'$ and $\Pi^*$ (sharing a face with $e'$ and $e^*$, respectively, and defined similarly as $\Pi$) previously, or we can place it in any edge of such contour paths. We have that we did not  place a dummy in a same edge in $G^1$ and $G^2$. In fact, $\Pi'$ and $\Pi^*$ are disjoint from $\Pi$.

The 2-1-placement can be handled symmetrically, while the 1-1-placement consider the placement we did in the 1-2-placement on $G^1$ symmetrically in $G^2$. Again, the contour paths where we eventually place two more dummies are different. As we observed before, we have 3 edges and for dummies to place and now we showed how to handle optimally the three possible strategies. Each procedure takes linear time, since we only consider one bad cycle (either the 2-legged cycle $C_{bad}$ or the degenerate 3-legged cycle $C_{bad}'$). We try all the combinations of the placements and we chose the one with the minimum number of dummies.

\smallskip
\noindent {\sffamily\bfseries (iii) External-Demanding Case:
  $D_{\ext}(G)\not = \emptyset$.} As there is no intersecting
cycle, since the cases we are considering are mutually exclusive, we
have that the maximal 3-legged cycles are independent. Hence, we can
use the nice-placement so that now every 3-legged cycle is not a bad
cycle. We placed exactly $q'=2|D(G)|+|D_{\exp}(G)|$ dummies. We have
now to see what happens in the external face $f_{\ext}$.  By \cref{app-le:cubic-unbentT}(2) we placed a dummy in each of $G^1$ $G^2$ for every cycle in $D_{\ext}(G)\setminus D_{\short}(G)$ and just one dummy (either in $G^1$ or $G^2$) for the cycles in $D_{\short}(G)$. The remaining dummies so that the external face has at least four vertices can be placed in other edges in $G^1$ and $G^2$. We have that by \cref{app-le:cubic-unbentT} and by \cref{app-le:cubic-lower3} we placed $\tbn(G)$ dummies. Also, by \cref{app-le:cubic-only_intersecting} and by  \cref{app-le:cubic-unbentT}  there is no bad cycle.
\end{proof}

Since a triconnected planar graph with $n$ vertices has $O(n)$ planar embeddings, an immediate consequence of~\cref{app-th:minbend-cubic} is the following.
\begin{corollary}
    Let $G$ be a planar triconnected cubic with $n$ vertices. There exists an $O(n^2)$-time algorithm that computes an unbent collection of $G$ with $\tbn(G)$ bends. Also, the collection has size two which is optimal.
\end{corollary}

\end{document}